\documentclass[lettersize,journal]{IEEEtran}
\IEEEoverridecommandlockouts
\usepackage{amsmath,amsfonts,amsthm,amssymb}
\usepackage{array}
\usepackage{textcomp}
\usepackage{subfigure}
\usepackage{float}
\usepackage{enumerate}
\usepackage{tabularx}
\usepackage{stfloats}
\usepackage{url}
\usepackage{hyperref}
\usepackage{xcolor,soul,framed} 
\usepackage{pstricks}
\usepackage{verbatim}
\usepackage{graphicx}
\usepackage{epstopdf}
\usepackage{booktabs}
\usepackage{cite}
\usepackage[T1]{fontenc}
\usepackage{tikz}
\usepackage{algorithm}
\usepackage{algpseudocode}
\usepackage{setspace}
\usepackage{booktabs}
\usepackage{makecell}
\usepackage{multirow}

\allowdisplaybreaks
\begin{document}
\newtheorem{Remark}{Remark}
\newtheorem{remark}{Remark}
\newtheorem{thm}{Theorem}
\renewcommand{\algorithmicrequire}{\textbf{Input:}} 
\renewcommand{\algorithmicensure}{\textbf{Output:}}

\title{Design and Optimization on Successive RIS-assisted Multi-hop Wireless Communications
\thanks{This work is supported by the National Natural Science Foundation of China (NSFC) under Grants NO.12141107, and the Interdisciplinary Research Program of HUST (2023JCYJ012).}
}

%
\author{Rujing~Xiong, Jialong~Lu, Jianan~Zhang,~\IEEEmembership{Student Member,~IEEE,}
		Minggang~Liu,~\IEEEmembership{Member,~IEEE,}
		Xuehui~Dong,~\IEEEmembership{Student Member,~IEEE,}
		Tiebin~Mi,~\IEEEmembership{Member,~IEEE,}
        Robert~Caiming~Qiu,~\IEEEmembership{Fellow,~IEEE,}
\thanks{R.~Xiong et al. are with the School of Electronic Information and Communications, Huazhong University of Science and Technology, Wuhan 430074, China (e-mail: rujing@hust.edu.cn).}

}
\maketitle



\maketitle

\begin{abstract}
As an emerging wireless communication technology, reconfigurable intelligent surface (RIS) has become a basic choice for providing signal coverage services in scenarios with dense obstacles or long tunnels through multi-hop configurations. Conventional works of literature mainly focus on alternating optimization or single-beam calculation in RIS phase configuration, which is limited in considering energy efficiency, and often suffers from inaccurate channel state information (CSI), poor convergence, and high computational complexity. This paper addresses the design and optimization challenges for successive RIS-assisted multi-hop systems. Specifically, we establish a general model for multi-hop communication based on the relationship between the input and output electric fields within each RIS. Meanwhile, we derive the half-power beamwidth of the RIS-reflected beams, considering the beam direction. Leveraging these models and derivations, we propose deployment optimization and beam optimization strategies for multi-hop systems, which feature high aperture efficiency and significant gains in signal power. Simulation and prototype experiment results validate the effectiveness and superiority of the proposed systems and methods.

\end{abstract}


\section{Introduction}
\IEEEPARstart{T}{he} Reconfigurable intelligent surface (RIS) has emerged as a promising technique for wireless communication networks~\cite{cui2014coding, basar2019wireless,di2020smart}. By dynamically adjusting the reflection phase shifts of numerous passive units, RIS enables flexible wireless channel control and configuration, thereby significantly improving wireless signal transmission rates and reliability. In particular, RIS is a digitally controlled metasurface composed of a large number of cost-effective, well-designed passive reflecting units, each capable of independently manipulating the characteristics of incident electromagnetic waves, including phase, amplitude, and polarization. Due to its great potential, RIS has been widely investigated in the scope of wireless communications~\cite{wu2024intelligent}.

Despite the extensive literature on the design and optimization of various RIS-assisted wireless systems, prior works have primarily focused on enhancing wireless links through single signal reflection using one or multiple RISs~\cite{you2020channel,pan2020multicell,yang2021energy,nguyen2023leveraging,xiong2023fairbeamallocationsreconfigurable}. This approach may be inadequate to improve wireless signal quality under some adverse propagation conditions, such as indoor/outdoor environments with dense blockages or obstructions. Additionally, the studies indicated that an RIS should be deployed near the base station (BS) or user equipment (UE) for enhanced communication performance~\cite{wu2019intelligent,zhang2021intelligent,wu2021intelligent}. This requirement further complicates the deployment of RIS.

To address the above challenges, cooperative multi-hop reflections are developed~\cite{huang2021multi,liang2022multi,mei2022multi,abou2024intelligent,ma2023multi}. In~\cite{huang2021multi}, the authors studied the multi-reflection between a BS and multiple users, jointly designing the active and passive beamforming at BS and RISs based on deep reinforcement learning. The authors in~\cite{liang2022multi} investigated a multi-route multi-hop cascaded RIS-aided system, with the objective of maximizing the sum rate for multiple users. In ~\cite{mei2020cooperative,mei2022multi}, the optimal reflection path for multi-RIS beam routing problem is studied. Additionally, the authors in~\cite{ma2023multi} optimized the multi-reflection paths for interference management and throughput performance boosting in multi-hop multi-RIS cooperative communications. \cite{abou2024intelligent} explored the integration of multi-hop RISs and open radio access networks, in which the best RIS-assisted path was selected for the overall network sum-rate maximization. 

Multi-hop RIS plays a crucial role in facilitating communications within complex terrain and extensive urban environments~\cite{mei2020cooperative}. Compared to the single-reflection link, the multireflection link in general provides more degrees of freedom to bypass dense and scattered obstacles in complex environments. Moreover, the multireflection links offer higher beamforming gains, which can counteract the product-distance path loss that also increases with the number of RIS reflections~\cite{mei2022intelligent}.

In the current literature on phase configuration within multi-hop communications, the methods can be primarily classified into two categories. The first considers the channels between multiple RISs in a cascaded manner, and achieving the optimization objective through alternating optimizations~\cite{zhang2021rate,shen2023joint,liang2023energy,niu2022double}. These methods have two significant drawbacks. One is that RISs are passive reflectors with a large number of units, rendering the acquisition of channel state information between RISs particularly challenging. Additionally, as the number of RISs increases, the expression to solve optimization problems for each RIS becomes increasingly complex, the iterations are highly prone to converging to local suboptimal solutions, and the computational complexity escalates sharply.

Another feasible method involves considering RIS one by one, in which the phase configuration of each RIS can be determined efficiently according to the channel state information of inner RIS~\cite{han2020cooperative,mei2020cooperative,mei2021cooperative}, or the relative positions of devices in the communication link~\cite{ma2023multi}. This phase configuration method is frequently employed in optimal routing-path selection within multi-RIS networks.

For RIS-assisted communications, different feed source positions always lead to variations in the field distribution on the RIS surfaces, which in turn cause changes in directional gain. Our previous experimental results have also observed this phenomenon~\cite{xiong2023ris}. According to radiation theory in array antennas~\cite{nayeri2018reflectarray,balanis2016antenna}, this variability can be attributed to the aperture efficiency of RIS. As a passive reflector array, the magnitude and uniformity of the electromagnetic wave incident on the RIS's aperture determine the efficiency and directional power gain it can provide. These characteristics have mostly been ignored in existing studies on RIS phase configuration.

In this paper, we take into account the aperture efficiency and beamwidth of the reflecting beam and derive a closed-form expression for beamwidth. Based on the derivations, we and propose strategies for the deployment and beam optimization in multi-hop RIS-assisted communications. In beam radiation design, we employ our previously proposed MA algorithm to achieve multiple beams, whose computational efficiency, optimality, and robustness have been demonstrated in our earlier work~\cite{xiong2023fairbeamallocationsreconfigurable}. The proposed multi-hop system features high aperture efficiency and significant power gains.

\subsection{Contributions}
The main contributions of this paper can be summarized as follows:
\begin{itemize}
\item \textbf{Generalized signal model for multi-hop RIS-assisted Communication}. We construct a system model that considers the propagation of signals between the base station (BS) and user equipment (UE) through successive multi-hop RIS configurations. The multiple RIS-reflection channel components among the propagation are categorized into two classes. We characterize the relationship between the input and output electric fields among RISs within each class and integrate them to establish a general multi-hop signal model. With this model, the electric field at each point can be expressed analytically.
\item \textbf{Derivation of the closed-form expression for the beamwidth of RIS-reflected beams}. As a reflectarray antenna, we derive the array factor incorporating the beam direction of RIS reflected. Based on that, we further present the closed-form expression for the half-power beamwidth (HPBW) of the reflecting beams, which closely correlates with the number of units (or aperture size) of RIS and the reflecting beam direction.
\item \textbf{Optimization strategy on RIS deployment and beam design}. Leveraging the derived beamwidth expressions and signal model, we propose optimization methods for RIS deployment and beam design, which eliminates the need for CSI estimation and prioritize high aperture efficiency. The derived results and proposed optimization strategies are validated through simulation and prototype experiments. In the prototype test, we demonstrate that employing optimized multi-beam can increase signal power gain by approximately \textbf{10 dB} compared to widely used single-beam approaches.
\end{itemize}

\subsection{Outline}
The remainder of the paper is organized as follows. In Section~\ref{Section2}, we present the modeling of successive RIS-assisted multi-hop communications. Section~\ref{Section3} analyzes the radiation of RIS, mainly on the aperture efficiency and array factor. Further, the closed-form expression of beamwidth for RIS-reflected beam is derived. In Section~\ref{Section4}, we proposed the deployment and beam optimization strategies for enhanced multi-hop RIS-assisted communication. Section~\ref{Section5} is dedicated to the performance evaluations of the derived results and proposed methods through numerical and prototype experiments. Finally, this paper is concluded in Section~\ref{Section6}.


\subsection{Notations}
Notations in this paper are defined as follows. The imaginary unit is indicated by $j$. The magnitude and complex components of a complex number are represented by $|\cdot|$ and $\mathcal{\Im}(\cdot)$, respectively. ${\rm arg}(\cdot)$ stands for the angle function. Unless explicitly specified, lower and upper case bold letters denote vectors and matrices. The conjugate transpose and transpose of~$\mathbf{A}$ are represented as $\mathbf{A}^H$ and $\mathbf{A}^T$, respectively. 

\section{System Model}\label{Section2}

In this section, we consider the scenario of the multi-hop RIS-assisted downlink communications, as depicted in Fig.~\ref{Scenarios}. For modeling, we categorize the propagations in multi-hop communications into two types: BS-RIS-UE, where the transmitter is the BS and considered a point source; and RIS-RIS-UE, where the signal source is the reflected RIS. These two types of signal models can be linked through the relationship between the incident field and the reflected field on RISs, thereby we can establish a general model for multi-hop RIS-assisted communication.

\subsection{BS-RIS-UE}

\begin{figure}
  \centering
  \includegraphics[width=.8\linewidth]{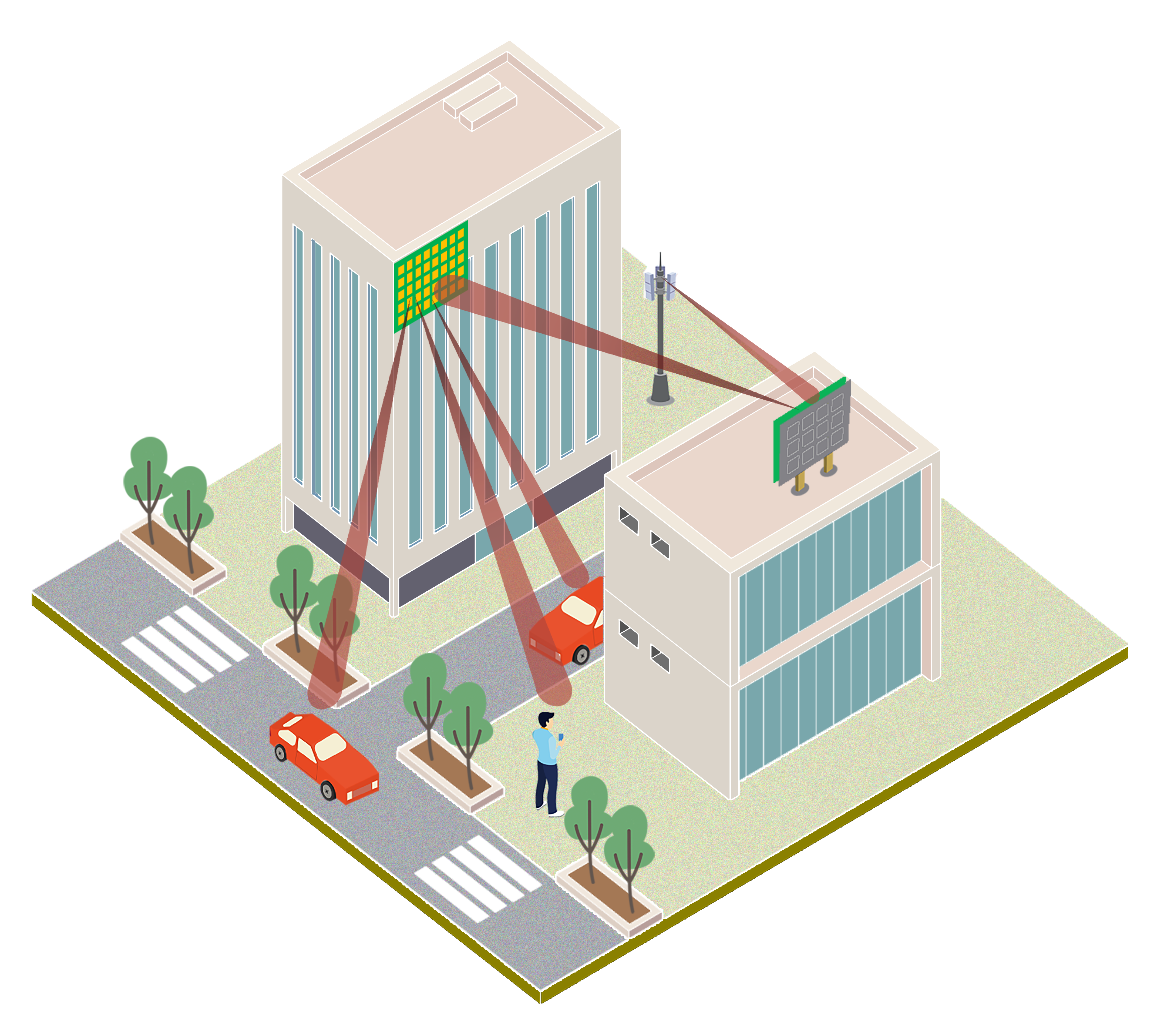}
  \caption{RIS-assistanted multi-hop wireless communications.}
  \label{Scenarios}
\end{figure}

For our analysis, we first concentrate on an arbitrary RIS situated on the $xoy$-plane, comprising multiple units located at $\mathbf{p}_n = [x_n, \ y_n, \ z_n ]^T$, $n=1, \ldots, N$. The RIS is illuminated by a single incident electromagnetic wave (EM) wave originating from a point source location $(r^\text{i}, \theta^\text{i}, \phi^\text{i})$. For clarity, $r^\text{i}$ denotes the distance from the source to the anchor of the RIS, such as the geometrical center for a planar array or the left end for a uniform linear array, while $r^\text{i} (n)$ represents the distance from the source to the $n$-th unit at $\mathbf{p}_n$, $ (\theta^\text{i}, \phi^\text{i})$ denotes the elevation and azimuth angles, as shown in Fig.~\ref{unit}

\begin{figure}
  \centering
  \includegraphics[width=1\linewidth]{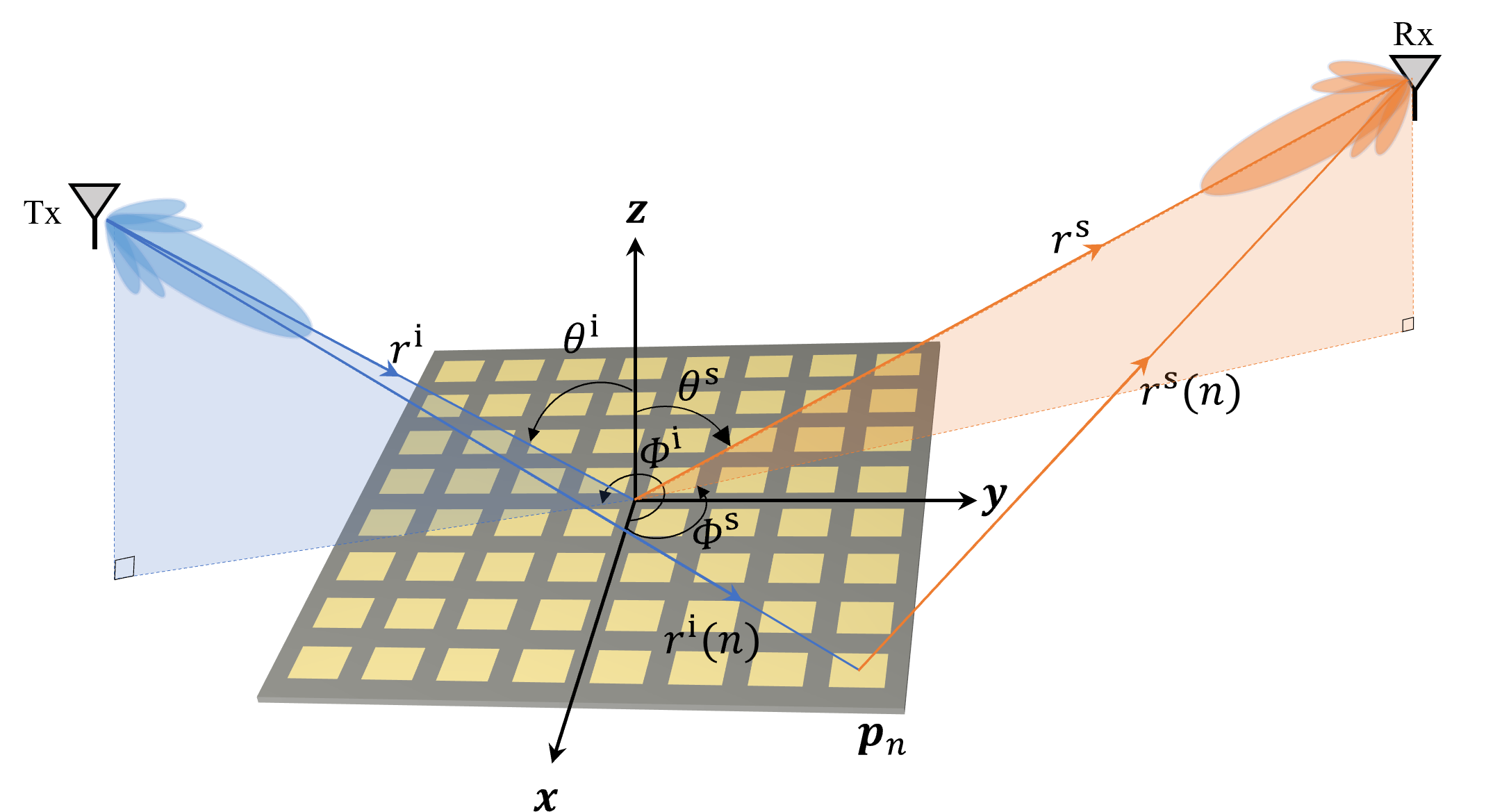}
  \caption{The RIS is illuminated by a single incident EM wave from ($ r^\text{i}, \theta^\text{i}, \phi^\text{i}$).}
  \label{unit}
\end{figure}

In the case of a point source, radiation propagates radially. As the EM wave propagates toward the unit at $\mathbf{p}_n$, the attenuation behavior is characterized by the factor $e^{ - j 2 \pi r^\text{i} (n) / \lambda}/ r^\text{i} (n)$. Suppose the original electric field from the point source is $E^\text{i} ( r^\text{i}, \theta^\text{i}, \phi^\text{i} )$. The incident electric field at $\mathbf{p}_n$ could be represented as 
\begin{equation}
E^\text{i} (n) = E^\text{i} ( r^\text{i}, \theta^\text{i}, \phi^\text{i} ) e^{ - j 2 \pi r^\text{i} (n) / \lambda}/ r^\text{i} (n). 
\end{equation}

Similarly, if $E^\text{s} (n)$ denotes the scattered electric field of the unit at $\mathbf{p}_n$, then the electric field at the observation point is $E^\text{s} (n) e^{ - j 2 \pi r^\text{s} (n) / \lambda } / r^\text{s} (n)$. 

As indicated in our previous study~\cite{xiong2023fairbeamallocationsreconfigurable}, the electric field at the observation point along the propagation path through the $n$-th unit at $\mathbf{p}_n$ can be obtained by combining the above components, as
\begin{multline}\label{SISOBehavior}
  E^\text{s}_n ( r^\text{s}_n, \theta^\text{s}, \phi^\text{s} ) = \tau ( \theta^\text{s} (n), \phi^\text{s} (n) ; \theta^\text{i} (n), \phi^\text{i} (n) ) e^{j \omega_n} \\
  \frac{ e^{ - j 2 \pi r^\text{i} (n) / \lambda} e^{ - j 2 \pi r^\text{s} (n) / \lambda } }{ r^\text{i} (n) r^\text{s} (n) } E^\text{i} ( r^\text{i}, \theta^\text{i}, \phi^\text{i} ) . 
\end{multline}
Here, $\tau(\theta^\text{s}, \phi^\text{s}; \theta^\text{i}, \phi^\text{i})$ is the scattering pattern of an isolated unit, which is dependent on both incident angle $(\theta^\text{s}, \phi^\text{s})$ and scattered angle $(\theta^\text{s}, \phi^\text{s})$. $\omega_n$ represents the phase configuration of the $n$-th unit within the RIS. 

In isotropic scattering, incident EM waves uniformly scatter in all directions over the hemisphere of reflection, irrespective of the angle of incidence or observation. Consequently, the unit's scattering pattern is simplified to $\tau_n ( \theta^\text{s}, \phi^\text{s}; \theta^\text{i}, \phi^\text{i} ) = \tau$. The electric field at the observation point is determined by the superposition of individual fields scattered by $N$ units. As
\begin{multline}\label{Model_SISO}
  E^\text{s} ( r^\text{s}, \theta^\text{s}, \phi^\text{s} ) \\ = E^\text{i} ( r^\text{i}, \theta^\text{i}, \phi^\text{i} ) \sum_{n=1}^{N} \tau 
  e^{j \omega_n} \frac{ e^{ - j 2 \pi r^\text{i} (n) / \lambda} e^{ - j 2 \pi r^\text{s} (n) / \lambda } }{ r^\text{i} (n)  r^\text{s} (n) }   . 
\end{multline}
This expression can be written in matrix forms, as 
\begin{equation}\label{Model_SISO_M}
\begin{aligned}
E^\text{s} ( r^\text{s}, \theta^\text{s}, \phi^\text{s} ) & = 
\tau
 \begin{bmatrix}
            \frac{ e^{-j 2 \pi r^\text{s}(1) / \lambda }}{ r^\text{s}(1) } &   \cdots     &  \frac{ e^{-j 2 \pi r^\text{s}(N) / \lambda }}{ r^\text{s}(N) } 
      \end{bmatrix} \\
 &\begin{bmatrix}
            e^{j \omega_1} &        &  0 \\
                               & \ddots &    \\
            0                  &        & e^{j \omega_N}
      \end{bmatrix} 
       \begin{bmatrix}
         \frac{ e^{-j 2 \pi r^\text{i}(1) / \lambda }}{ r^\text{i}(1) }                                                    \\
        \vdots                                                                            \\
        \frac{ e^{-j 2 \pi r^\text{i}(N) / \lambda }}{ r^\text{i}(N) }    
      \end{bmatrix} 
        E^\text i (r^\text{i}, \theta^\text{i}, \phi^\text{i}) .
\end{aligned}
\end{equation}

This signal model characterizes the input-output behavior of the RIS, enabling us to calculate the required RIS phase configuration based on the desired electric field. In other words, for a given RIS phase configuration, one only needs to determine the elevation and azimuth angles of the observation point relative to the RIS to calculate the electric field at an observation location.

\subsection{RIS-RIS-UE}

A notable feature in multi-hop RIS setups is that, except for the first RIS, all other RIS nodes are illuminated by the reflected signal from its former RIS. This means that these illuminations are often directional and non-uniform, especially in the near-fields. In this section, we consider the field distribution on RIS to model the signal propagation.

Suppose totaly $K$ RISs are employed, the $(k-1)$-th RIS (denoted as RIS$\{k-1\}$) employs $N$ effective units reflecting the signal beam to illuminate the $k$-th RIS (denoted as RIS$\{k\}$, with $M$ units), where $k\geq 2$. For RIS$\{k\}$, the incident electric field at the $m$-th unit (coordinates $\mathbf{p}_m = [x_m, \ y_m, \ z_m ]^T$, $m=1, \ldots, M$) can be calculate through the scattered field of RIS$\{k-1\}$ at $\mathbf{p}_m$, which can be expressed as 
\begin{equation}\label{E:IS}
E^\text {i}_{k} (m) =  E^\text{s}_{k-1} (r^\text{s}_{k-1}(m), \theta^\text{s}_{k-1},\phi^\text{s}_{k-1}),
\end{equation} 
where $r^\text{s}_{k-1}(m)$ represent the distance from the $m$-th unit on RIS$\{k\}$ to the reference center point on RIS$\{k-1\}$.

Furthermore, the electric field observed at a specific point after reflection by the hop-node RIS$\{k\}$ can be written as 
\begin{equation}\label{RIS-RIS-UE}
\begin{aligned}
&E_k^{\mathrm{s}}(r^\text{s}_k, \theta^\text{s}_k, \phi^\text{s}_k)\\
& = \sum_{m=1}^M \tau  e^{j \omega_{m}} E^\text {i}_{k} (m)\frac{e^{-j 2 \pi r^\text{s}_k(m )/ \lambda}}{r^\text{s}_k(m )}\\
& = \sum_{m=1}^M \tau  e^{j \omega_{m}} E^\text{s}_{k-1} (r^\text{s}_{k-1}(m), \theta^\text{s}_{k-1},\phi^\text{s}_{k-1})
 \frac{e^{-j 2 \pi r^\text{s}_k(m )/ \lambda}}{r^\text{s}_k(m )}.
\end{aligned}
\end{equation}
It should be noted that this expression is also an iterative calculation about $E^\text{s}_{k-1} (r^\text{s}_{k-1}(m), \theta^\text{s}_{k-1},\phi^\text{s}_{k-1})$, with $2\leq k\leq K$. The observed electric field at a certain point depends on the spatial scattered field of the all $K$ RISs.

\subsection{General Model}
While $k =1$, for each $m$, $E^\text{s}_{k-1} (r^\text{s}_{k-1}(m), \theta^\text{s}_{k-1},\phi^\text{s}_{k-1})$ can be obtained through~\eqref{Model_SISO}. Thus, by combining equation \eqref{Model_SISO}, \eqref{E:IS}, and \eqref{RIS-RIS-UE}, we can construct a general cascaded model for signal transmission from the BS to a certain UE, with $K$ multi-hop RISs employed.

For simplicity and clarity, we consider a classic two-hop signal model in analysis to reveal this propagation of wireless signals along multi-hop RISs. As illustrated in Fig,~\ref{TwoHop}, the signal propagates along the BS-RIS$\{1\}$-RIS$\{2\}$-UE path, and no direct links between the BS and RIS$\{2\}$. The number of units on RIS$\{1\}$ and RIS$\{2\}$ is $N$ and $M$, respectively.
\begin{figure*}
  \centering
  \includegraphics[width=0.8\linewidth]{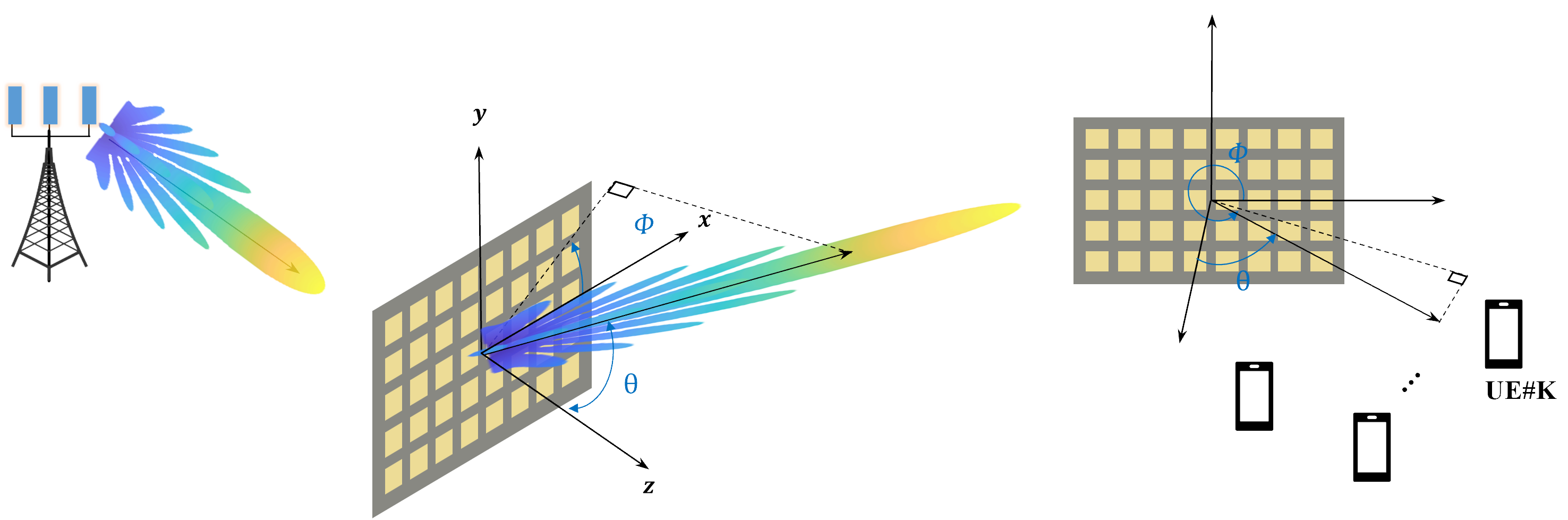}
  \caption{Two-Hop RIS-assisted wireless communications.}
  \label{TwoHop}
\end{figure*}
After two hops, the scattered field observed at the UE point can be expressed as
\begin{equation}
\begin{aligned}
  E^\text{s} (p ) &\\ 
  =&\sum_{m=1}^M \tau  e^{j \omega_m} E^\text{i}_2 (m) \frac{e^{-j 2 \pi r^\text{s}_2(m )/ \lambda}}{r^\text{s}_2(m )} \\
  =&\sum_{m=1}^M \tau  e^{j \omega_m} E^\text{s} (r^\text{s}(m), \theta^\text{s},\phi^\text{s}) \frac{e^{-j 2 \pi r^\text{s}_2(m )/ \lambda}}{r^\text{s}_2(m )} 
\end{aligned}
\end{equation}
For each $1\leq m\leq M$, $E^\text{s}_{k-1} (r^\text{s}_{k-1}(m), \theta^\text{s}_{k-1},\phi^\text{s}_{k-1})$ can be calculated through~\eqref{Model_SISO}.
Thus, the equation holds
\begin{multline}
 E^\text{s} ( p ) 
  =\sum_{m=1}^M \tau  e^{j \omega_{m}}  \frac{e^{-j 2 \pi r^\text{s}_2(m )/ \lambda}}{r^\text{s}_2(m )}E^\text{i} ( r^\text{i}, \theta^\text{i}, \phi^\text{i} )
   \sum_{n=1}^{N} \tau  e^{j \omega_{n}} \\
   \frac{ e^{ - j 2 \pi r^\text{i} (n) / \lambda} e^{ - j 2 \pi r^\text{s} (m,n) / \lambda } }{ r^\text{i} (n)  r^\text{s} (m,n) }.
\end{multline}
Here, $r^\text{s} (m,n)$ represents the distance between the $m$-th unit on RIS$\{2\}$ and the $n$-th unit on RIS$\{1\}$. 

This equation can be written in matrix forms as in \eqref{E:2Hop}. We have the canonical linear representation to describe the input/output behaviors of multi-hop RIS-assisted communications. The advantage lies in employing a simple system of linear equations to describe input/output field, making it suitable for analyzing and optimizing the performance of RIS-aided systems.

\begin{figure*}[!htbp]
  \begin{equation}\label{E:2Hop}
    \begin{aligned}    
        E^\text{s} (p) 
      = & \tau^2 
 \begin{bmatrix}
            \frac{ e^{-j 2 \pi r^\text{s}_2(1 ) / \lambda }}{ r^\text{s}_2(1 ) } &   \cdots     &  \frac{ e^{-j 2 \pi r^\text{s}_2(M ) / \lambda }}{ r^\text{s}_2(M ) } 
      \end{bmatrix} 
   \begin{bmatrix}
            e^{j \omega_1} &        &  0 \\
                               & \ddots &    \\
            0                  &        & e^{j \omega_M}
      \end{bmatrix} 
      \begin{bmatrix}
        \frac{ e^{-j 2 \pi r^\text{s}(1,1) / \lambda }}{ r^\text{s}(1,1) }      & \cdots   &  \frac{ e^{-j 2 \pi r^\text{s}(1,N) / \lambda }}{ r^\text{s}(1,N) } \\
        \vdots & \ddots & \vdots                                          \\
        \frac{ e^{-j 2 \pi r^\text{s}(M,1) / \lambda }}{ r^\text{s}(M,1) }       & \cdots & \frac{ e^{-j 2 \pi r^\text{s}(M,N) / \lambda }}{ r^\text{s}(M,N) }  \\
      \end{bmatrix} 
      \begin{bmatrix}
            e^{j \omega_1} &        &  0 \\
                               & \ddots &    \\
            0                  &        & e^{j \omega_N}
      \end{bmatrix} \\
      & \begin{bmatrix}
         \frac{ e^{-j 2 \pi r^\text{i}(1) / \lambda }}{ r^\text{i}(1) }                                                    \\
        \vdots                                                                            \\
        \frac{ e^{-j 2 \pi r^\text{i}(N) / \lambda }}{ r^\text{i}(N) }    
      \end{bmatrix} 
        E^\text i (r^\text{i}, \theta^\text{i}, \phi^\text{i}) .
    \end{aligned}
  \end{equation}
  \medskip
  \hrule
\end{figure*}


\section{Radiation Analysis of the RIS}\label{Section3}

One of the primary objectives of successive multi-hop RIS is to enhance the received signal power at the UEs. Traditional multi-hop communications have almost overlooked the transmission efficiency between RISs, leading to exponential signal energy loss during the multi-hop process. In this section, we take into account the aperture efficiency and half-power beamwidth (HPBW) analysis for the design and optimization of multi-hop systems, which are demonstrated to be crucial for signal improvement.
\subsection{Aperture Efficiency}

We first consider the aperture efficiency, a parameter that measures the ability of an antenna's physical aperture to convert available input power into radiated, significantly impacting the overall performance and signal strength of the antenna. In aperture-type antennas such as RISs, the physical aperture area $A_{\text p}$ determines the theoretical maximum directivity of the antenna, which is 
\begin{equation}
D_{\text max} =  \frac{4 \pi}{\lambda^2} A_{\text p}.
\end{equation}
The achievable gain of RIS can be computed as 
\begin{equation}
G=\varepsilon_{\mathrm{ap}} D_{\text max}=\varepsilon_{\mathrm{ap}} \frac{4 \pi}{\lambda^2} A_{\text p},
\end{equation}

where  $\varepsilon_{\text {ap }}$  is the aperture efficiency of the array and $ 0 \leq \varepsilon_{\text {ap }} \leq 1$, $A_{\text p}$ is the physical aperture area of RIS. Since the wavelength and physical aperture area are easily determined, the study of gain reduces to an analysis of aperture efficiency, which can be expressed as a product of sub-efficiencies.~\cite{stutzman2012antenna}:
\begin{equation}
\varepsilon_{\text {ap }} = e_{\text r}\varepsilon_{\text {s }}.
\end{equation}
Here
$e_{\text r}=$ illumination efficiency, $\varepsilon_{\text {s }}=$ spillover efficiency.

Regarding a RIS comprising $N$ units with $A$ rows and $B$ columns, the radiation efficiency $e_{\text r}$ is associated with the electric field distribution that the feed source irradiates onto the reflecting surface (RIS), which can be calculated through the following equation~\cite{nayeri2018reflectarray,zhang2005radar}.
\begin{equation}\label{Re}
e_{\text r}=\frac{\left|\sum_{a=1}^A \sum_{b=1}^B \left|E_{a b}\right|\right|^2}{A B \sum_{a=1}^A \sum_{b=1}^B\left|E_{a b}\right|^2},
\end{equation}

It can be observed that $e_{\text r}$ is maximized when the electric field $E_{ab}$ at each unit on the RIS is equal. In other words, the aperture field on the RIS is uniformly distributed.

Spillover efficiency $\varepsilon_{\text {s }}$ measures that portion of the feed pattern that is intercepted by the effective aperture of the RIS plate (and redirected through the aperture into the main beam) relative to the total energy\cite{stutzman2012antenna}. As

\begin{equation}\label{Se}
\varepsilon_s=\frac{\int_0^{2 \pi} \int_0^{\theta_0}\left|F_t\left(\theta_t, \phi_t\right)\right|^2 \sin \theta_t d \theta_t d \phi_t}{\int_0^{2 \pi} \int_0^{\pi/2}\left|F_t\left(\theta_t, \phi_t\right)\right|^2 \sin \theta_t d \theta_t d \phi_t},
\end{equation}
where $\theta_0$ is the angular aperture of the reflector (angle between the line from the feed source to the edge of the antenna aperture and the central axis of the antenna aperture), $F_t\left(\theta_t, \phi_t\right)$ represents the radiation pattern function of the feed source.
The spillover efficiency $\varepsilon_s$ is maximized when the numerator $\int_0^{2 \pi} \int_0^{\theta_0}\left|F_t\right|^2 \sin \theta_t d \theta_t d \phi_t$ and denominator $\int_0^{2 \pi} \int_0^{\pi/2}\left|F_t\right|^2 \sin \theta_t d \theta_t d \phi_t$ are equal, that is when the radiation beam illuminate solely on the RIS board.

\begin{figure}[H]\label{ApertureEfficiency}
  \centering
  \subfigure[]{
  \label{AE}
  \includegraphics[width=0.475\linewidth]{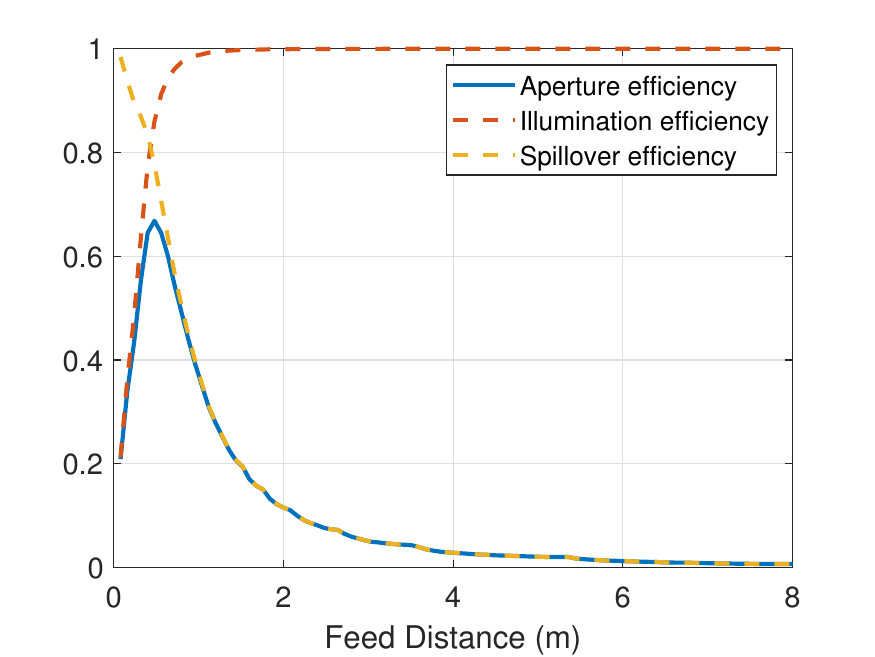}}
  \subfigure[]{
  \label{DG}
  \includegraphics[width=.475\linewidth]{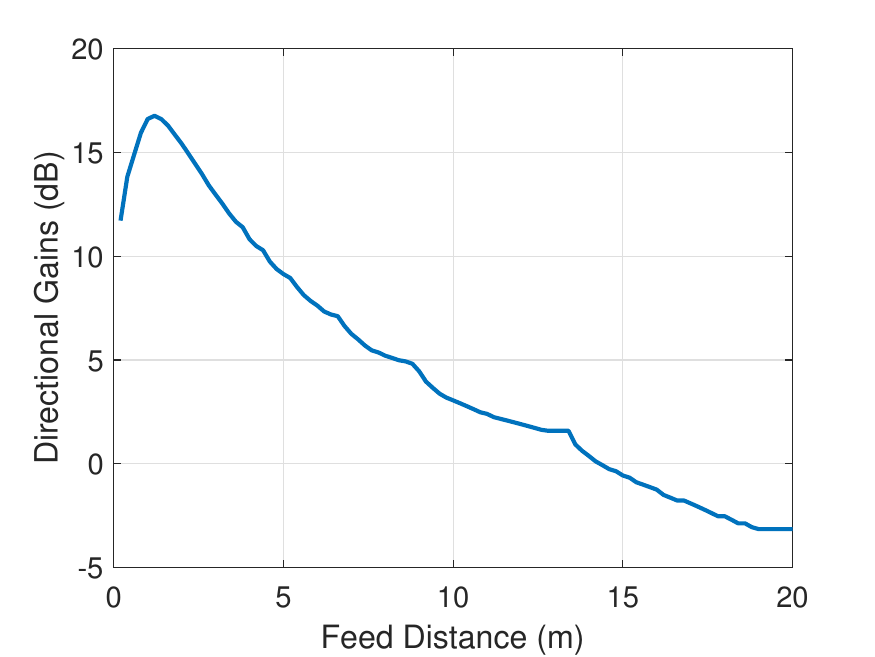}}
  \caption{Aperture efficiency and directional gains as a function of the distance between the feed source and the RIS.}
\end{figure}
\begin{figure*}[t]
  \centering
  \subfigure[]{
  \label{F2-1}
  \includegraphics[width=0.3\linewidth]{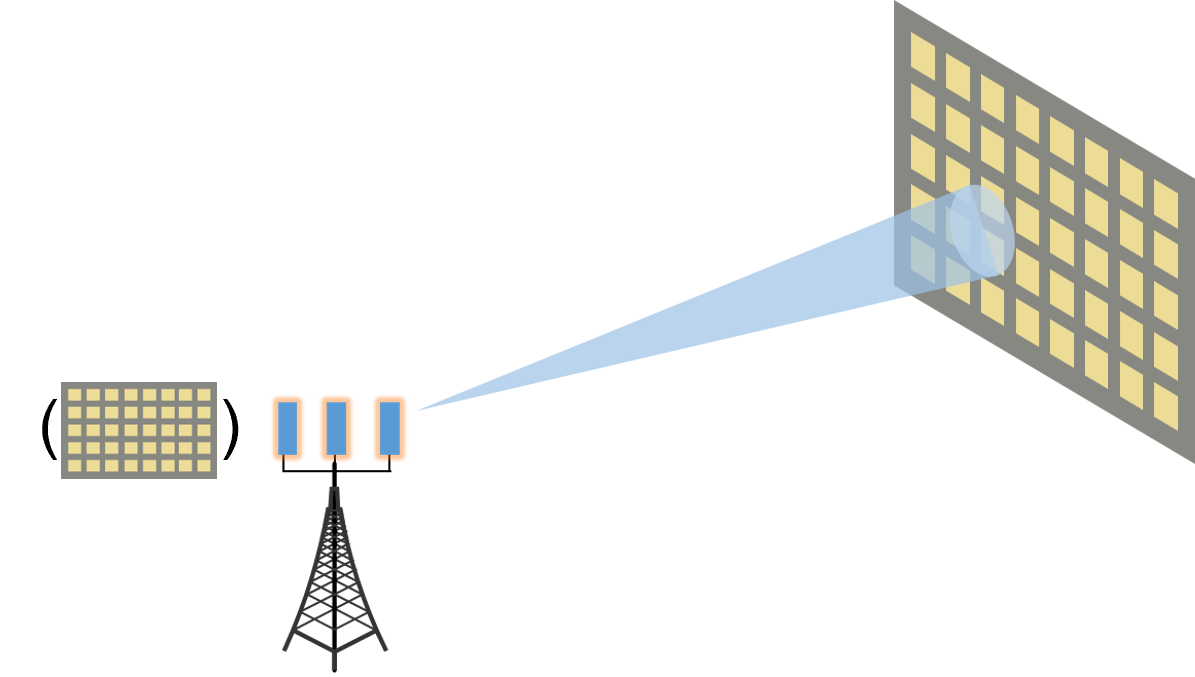}}
  \subfigure[]{
  \label{F2-2}
  \includegraphics[width=.3\linewidth]{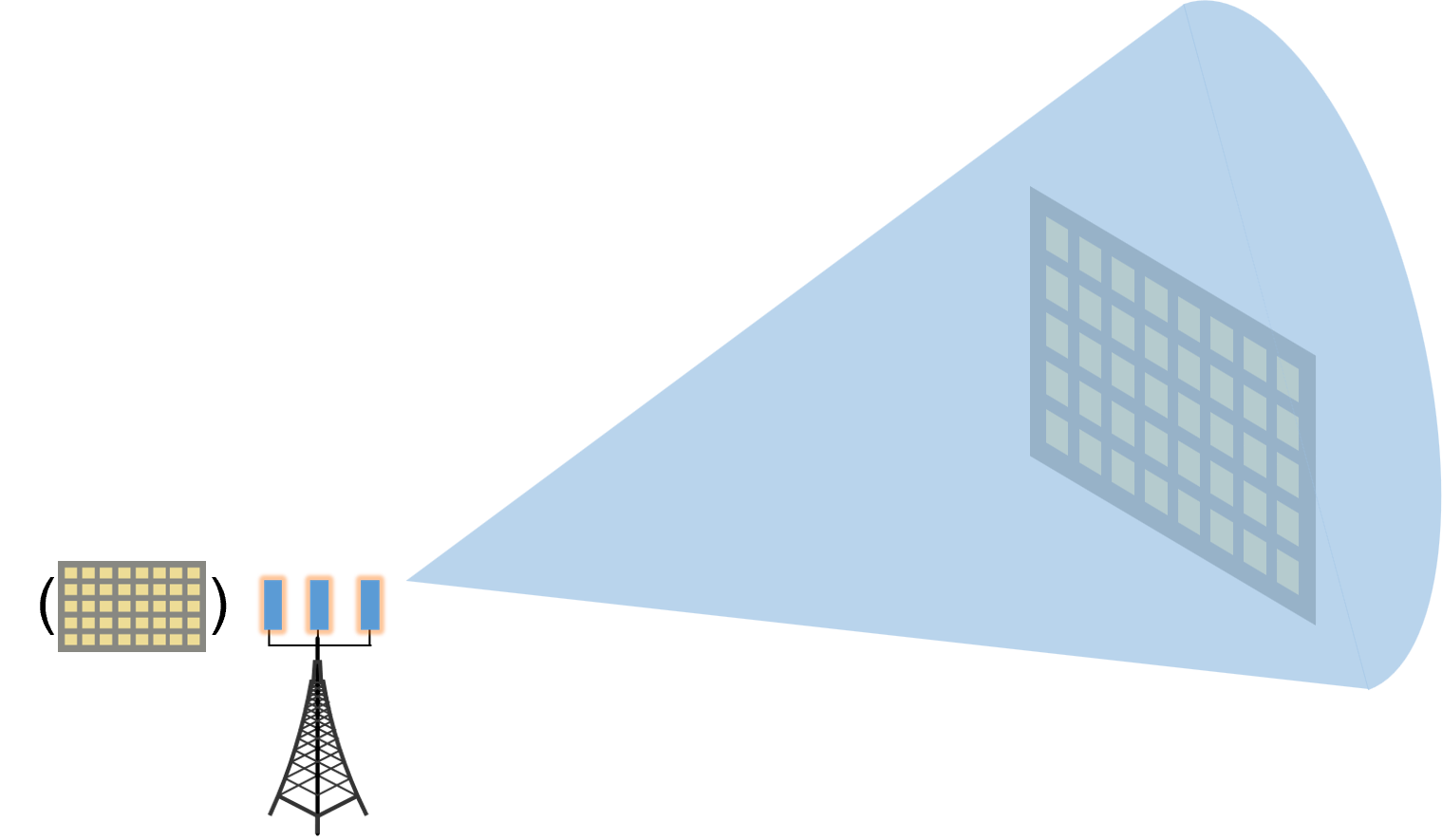}}
  \subfigure[]{
  \label{F2-3}
  \includegraphics[width=.3\linewidth]{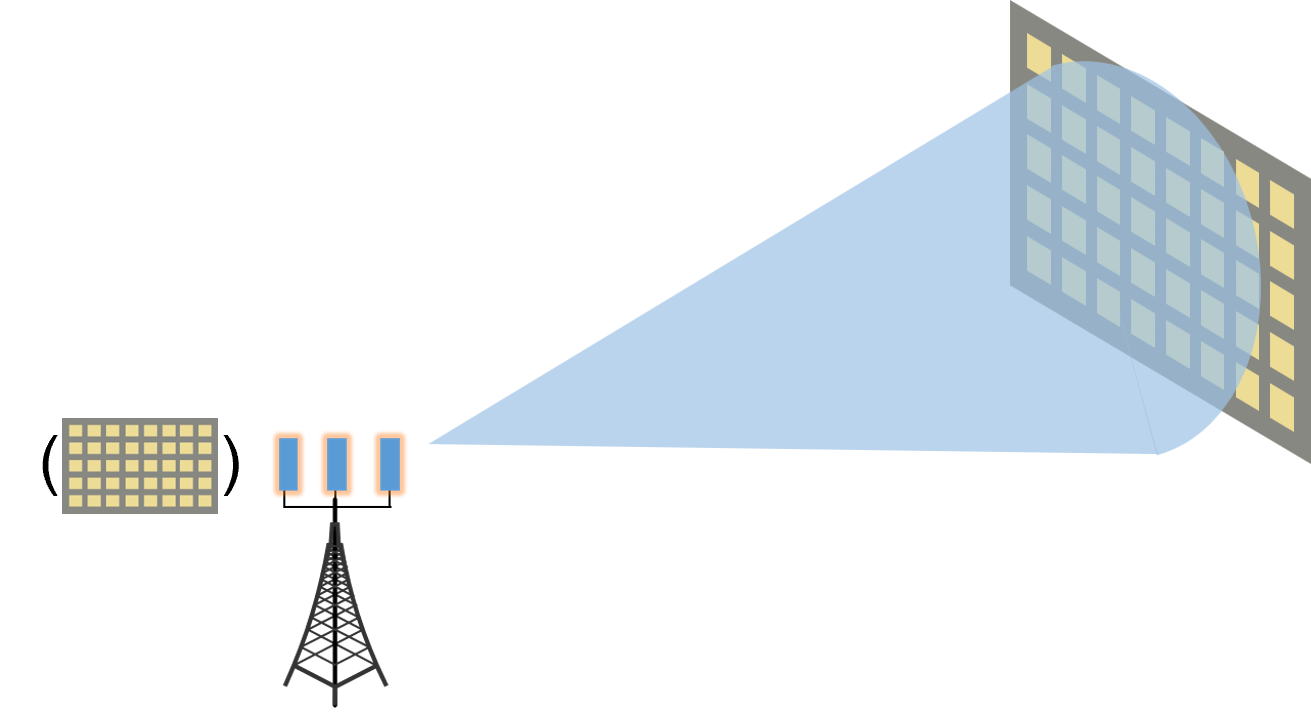}}
  \caption{The irradiation of the transmitting/reflecting beam. (a) Partial illumination. (b) Over illumination. (c) Exact illumination.}
  \label{ToRIS}
\end{figure*}

Recall that $G$ represents the directional gain achievable by the RIS acting as an aperture antenna. It's maximum is achieved when the aperture efficiency $\varepsilon_{\mathrm{ap}}$ peaks. Fig.\ref{ApertureEfficiency} illustrates the relationship between the distance from the feed source to the RIS and the efficiency and directional gain when a horn antenna is used as the feed source. It can be observed that as the distance increases, the illumination efficiency gradually increases, while the spillover efficiency gradually decreases. Meanwhile, the aperture efficiency, which is the product of these two factors, and the directionally correlated gain exhibit a trend of initially increasing and then decreasing.

Based on the equation\eqref{Re}, \eqref{Se}, and the antenna radiation theory~\cite{balanis2016antenna,tang2020wireless}, in communication links where RIS serves as a reflective feed source, one crucial approach to maximizing antenna gain $G$ can be designing a \textbf{precise beam} that \textbf{uniformly covers} the RIS surface area. As illustrated in Fig.~\ref{F2-3}.

\subsection{The Half-Power Beamwidth}

To quantify the signal coverage performance, the half-power beamwidth (HPBW) is one of the most important metrics~\cite{balanis2016antenna}. HPBW is conventionally defined as the beamwidth within which the signals's normalized radiated power is higher than -3 dB in a narrowband system. In this paper, we consider that there is a uniform electric field distribution within the 3-dB beamwidth coverage area.

With ideal aperture efficiency, the RIS can be modeled as a uniform amplitude and spacing array. The array factor can be derived by treating the units to be point sources. According to the pattern multiplication rule described in antenna theory~\cite{balanis2016antenna}, for an array composed of identical units, the total field can be synthesized by multiplying the array factor of the isotropic sources by the field of a single unit.

Assuming the position vector as follows,

\begin{enumerate}[\label={}]
    \item the transmitter (Tx):
    $$\mathbf{r}_1= \left[\cos \phi_1 \sin \theta_1,\sin \phi_1 \sin \theta_1, \cos \theta_1\right]^T,$$
    \item the receiver (Rx): $$\mathbf{r}_2=\left[\cos \phi_2 \sin \theta_2,\sin \phi_2 \sin \theta_2,\cos \theta_2\right]^T,$$
    \item the observation direction: $$\mathbf{r}=[\cos \phi \sin \theta,\sin \phi \sin \theta,\cos \theta]^T,$$
    \item the unit in the $a$-th row and $b$-th column: $$\mathbf{r}_{a b}=[(b-1) d,(a-1) d, 0]^T.$$
\end{enumerate}

Denote the phase of the ($a,b$)-unit on RIS as $\phi_{a b}$, as the main beam direction is achieved at ($\theta_2, \phi_2$), the phase configuration is
\begin{equation}
\phi_{a b}=-k\left(\mathbf{r}_1 * \mathbf{r}_{a b}^T + \mathbf{r}_2 * \mathbf{r}_{a b}^T\right),
\end{equation}
where $k = 2\pi/\lambda$ is the wave number.
The array factor of RIS can be given by
\begin{equation}\label{AF}
\mathrm{AF}(\theta, \phi)=\sum_{a=1}^A \sum_{b=1}^B  e^{j k\left(\mathbf{r} * \mathbf{r}_{ab}^T-\mathbf{r}_2 * \mathbf{r}_{ab}^T\right)},
\end{equation}
 which can also be rewritten in~\eqref{AF0}.
 \begin{figure*}
 \begin{equation}\label{AF0}
 \begin{aligned}
\mathrm{AF}(\theta, \phi)&= \sum_{a=1}^A \sum_{b=1}^B e^{j kd\left[((b-1) \cos \phi \sin \theta+(a-1) \sin \phi \sin \theta)-\left((b-1) \cos \phi_2 \sin \theta_2+(a-1) \sin \phi_2 \sin \theta_2\right)\right]} \\
&= \sum_{a=1}^A e^{j kd(a-1)\left(\sin \phi \sin \theta-\sin \phi_2 \sin \theta_2\right)} * \sum_{b=1}^B e^{j kd(b-1)\left(\cos \phi \sin \theta-\cos \phi_2 \sin \theta_2\right)}
\end{aligned}
\end{equation}
\end{figure*}
 
Denote $\Psi_1=kd\left(\sin \phi \sin \theta-\sin \phi_2 \sin \theta_2\right)$, and $\Psi_2=kd\left(\cos \phi \sin \theta-\cos \phi_2 \sin \theta_2\right)$, it holds
\begin{equation}\label{AF-11}
\mathrm{AF}(\theta, \phi)= \frac{\sin \left(\frac{A}{2} \Psi_1\right)}{\sin \left(\frac{1}{2} \Psi_1\right)} \frac{\sin \left(\frac{B}{2} \Psi_2\right)}{\sin \left(\frac{1}{2} \Psi_2\right)} e^{j\left(\frac{A-1}{2} \Psi_1+\frac{B-1}{2} \Psi_2\right)}
\end{equation}

If the reference point is the physical center of the array, the array factor of~\eqref{AF-11} reduces to
\begin{equation}\label{AF-12}
\mathrm{AF}(\theta, \phi)= \frac{\sin \left(\frac{A}{2} \Psi_1\right)}{\sin \left(\frac{1}{2} \Psi_1\right)} \frac{\sin \left(\frac{B}{2} \Psi_2\right)}{\sin \left(\frac{1}{2} \Psi_2\right)}. 
\end{equation}
It can be normalized as 
\begin{equation}\label{AF-12}
\mathrm{AF}(\theta, \phi)= \frac{1}{A}\frac{1}{B}\frac{\sin \left(\frac{A}{2} \Psi_1\right)}{\sin \left(\frac{1}{2} \Psi_1\right)} \frac{\sin \left(\frac{B}{2} \Psi_2\right)}{\sin \left(\frac{1}{2} \Psi_2\right)}. 
\end{equation}

The above expression can be approximated by 
\begin{equation}\label{Approxi}
\mathrm{AF}(\theta, \phi) \simeq\left(\frac{\sin \left(\frac{A}{2} \Psi_1\right)}{\frac{A}{2} \Psi_1} \frac{\sin \left(\frac{B}{2} \Psi_2\right)}{\frac{B}{2} \Psi_2}\right)
\end{equation}

With the array factor, we proceed to investigate the HPBW of the directional main beam. Considering the rotational symmetry of the beam, we take the plane containing the azimuth angle of the beam direction as the reference plane, i.e., $\phi=\phi_2$.

Substitute $\phi=\phi_2$ into $\Psi_1$ and $\Psi_2$, we hold the equations
\begin{equation}
\begin{aligned}
\Psi_1&=kd\sin \phi_2 (\sin \theta-\sin \theta_2),\\
\Psi_2&= kd  \cos \phi_2 (\sin \theta-\sin \theta_2).
\end{aligned}
\end{equation}

Denote $\Psi_3 = \frac{A}{2}\Psi_1$, and $\Psi_4 = \frac{B}{2}\Psi_2$, the factor array can be written as 
\begin{equation}
\mathrm{AF}(\theta, \phi) \simeq \frac{\sin \left(\Psi_3\right)}{\Psi_3} \frac{\sin \left(\Psi_4\right)}{\Psi_4}
\end{equation}

For the main beam in the plane at any azimuth angle $\phi_2$, the HPBW depends on the elevation angle $\theta$. To simplify and clarify, we illustrate this with $\phi_2=0$, where it holds
\begin{equation}\label{AF PLANAR}
\mathrm{AF}(\theta, \phi) \simeq  \frac{\sin \left(\Psi_4\right)}{\Psi_4}.
\end{equation}
The HPBW can be obtained through
\begin{equation}\label{HP}
\mathrm{HP} = |\theta^{*}_1-\theta^{*}_{2}|,
\end{equation}
where
\begin{equation}\label{E:3dB} 
\begin{aligned}
\theta^{*}_1&=\sin ^{-1}\left(  1.391 \times \frac{2}{B kd}+\sin \theta_{2}\right),\\
\theta^{*}_2&=\sin ^{-1}\left(  -1.391 \times \frac{2}{B kd}+\sin \theta_{2}\right).
\end{aligned}
\end{equation}
\begin{thm}\label{thm:error_est}
Given a reflecting angle $\theta_2$ of RIS, the half-power beamwidth can be obtained through $\mathrm{HP} = |\theta^{*}_1-\theta^{*}_{2}|$, where $\theta^{*}_1=\sin ^{-1}\left(  1.391 \times \frac{2}{B kd}+\sin \theta_{2}\right)$, and $\theta^{*}_2=\sin ^{-1}\left(  -1.391 \times \frac{2}{B kd}+\sin \theta_{2}\right)$.
\end{thm}

\begin{proof}
According to the AF approximation in~\eqref{AF PLANAR}, the 3-dB point for the array occurs when 
\begin{equation*}
\Psi_4=
\frac{B}{2} k d\left(\sin \theta^*-\sin \theta_{2}\right)= \pm 1.391,
\end{equation*}
we obtain
\begin{equation*}
\theta^{*}=\sin ^{-1}\left( \pm 1.391 \times \frac{2}{B kd}+\sin \theta_{2}\right).
\end{equation*}
Then 
\begin{equation*}
\mathrm{HP} = |\theta^{*}_1-\theta^{*}_{2}|.
\end{equation*}
\end{proof}
It is noteworthy that when the reflecting elevation angle is relatively large, the AF approximation function~\eqref{Approxi} tends to underestimate the actual value, resulting in a computed HPBW that is smaller than the actual one.

While the main beam direction is arbitrary, the HPBW can also be obtained by constructing and solving an equation through the 3-dB value point (see Appendix~\ref{AppB}).

For an $N$-unit uniformly distributed linear array with spacing $d$.
The array factor expression in \eqref{AF} can be reduced as 
\begin{equation}
\begin{aligned}
\mathrm{AF}(\theta)=\sum_{n=1}^N e^{j(n-1)k d (\sin \theta-\sin \theta_2)}
\end{aligned}
\end{equation}

which can be normalized and approximated as (see Appendix~\ref{AppA})
\begin{equation}
\mathrm{AF}(\theta) \simeq\left[ \frac{\sin \left(\frac{N}{2} \Psi_1\right)}{\frac{N}{2} \Psi_1}\right]
\end{equation}
where $\Psi_1 = k d (\sin \theta-\sin \theta_2)$.

The 3-dB point and beamwidth can also be obtained through~\eqref{HP} and~\eqref{E:3dB}.

The closed-form expressions~\eqref{HP} and ~\eqref{E:3dB} characterize the relationship between the main beam's beamwidth in a reflective array and the number of units, as well as the reflection direction, for the first time. In the subsequent section, we will optimize the deployment and beam design of RIS based on these findings.

\section{Optimization Strategy}\label{Section4}
The width of the beam reflected by the RIS is highly dependent on the location of the receiver. In this section, we optimize RIS-assisted communications from both deployment and beamforming perspectives. By leveraging the aperture efficiency theory and half-power beamwidth derived in the previous section, our objective is to enhance the performance of the multi-hop communications through the homogenization of aperture field distribution.

\begin{figure}
  \centering
  \includegraphics[width=0.9\linewidth]{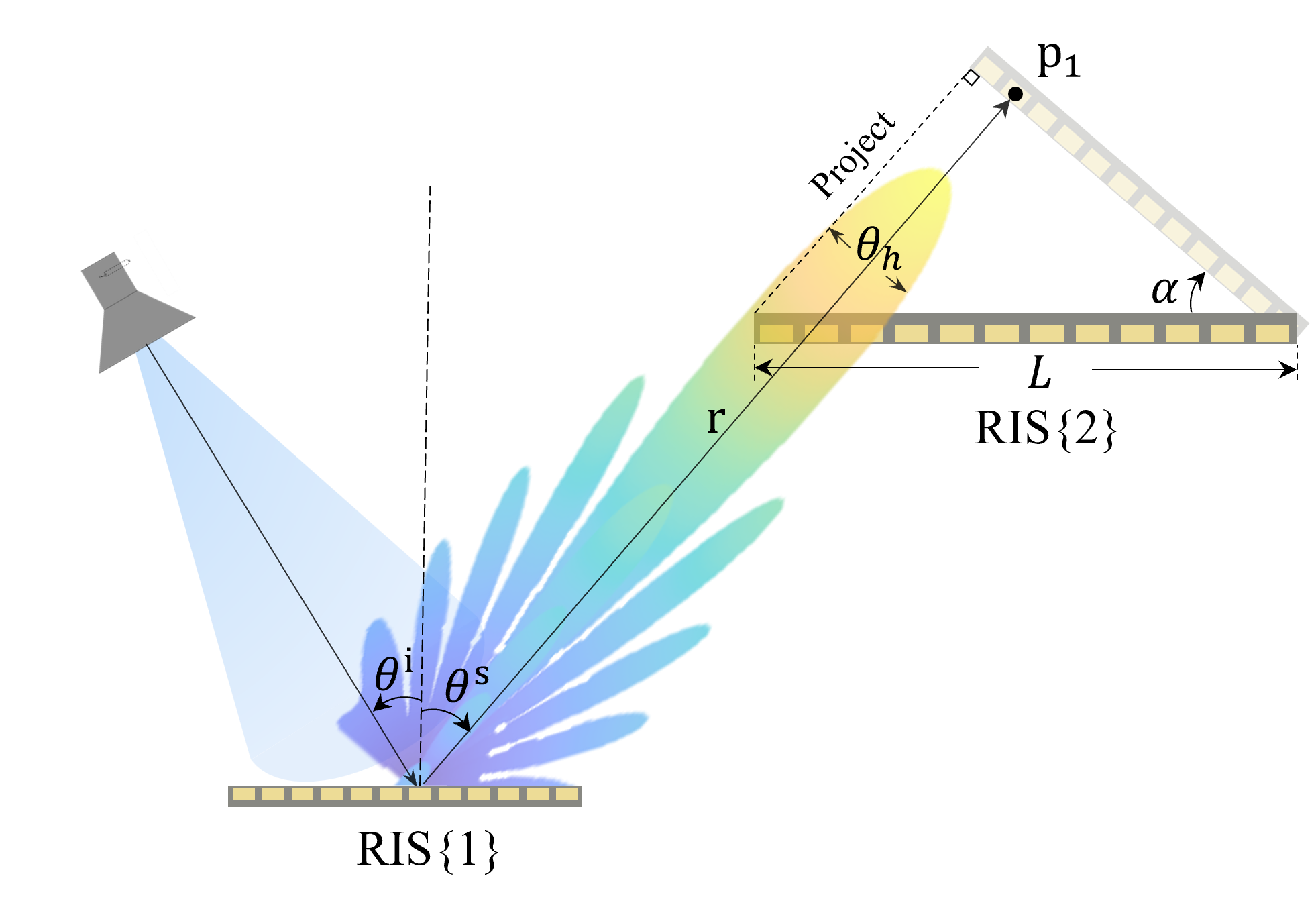}
  \caption{Beamwidth and effective aperture of RIS.}
  \label{Projection}
\end{figure}

\subsection{Deployment Optimization}

In successive multi-hop RIS-assisted communication, signals are typically point-to-point. In such single-input single-output scenarios, the phase configuration of the RIS can be directly and easily computed with the location of the transmitter and receiver~\cite{stutzman2012antenna,xiong2023fairbeamallocationsreconfigurable}.
One of the choices to maximize aperture efficiency and directional gains is to optimize the deployment of RISs according to the beamwidth.
Specifically, one can select an appropriate installation distance based on a fixed RIS aperture, or determine the required RIS aperture (the number of units) based on a fixed distance. For the sake of clarity, we employ a uniform linear array with $N$ units for analysis.

As shown in Fig.~\ref{Projection}, RIS$\{1\}$ reflecting the incident signals to RIS$\{2\}$. The effective aperture of RIS$\{2\}$ is $L\cos \alpha$.

\subsubsection{RIS position optimization}
For a fixed-aperture RIS, selecting an appropriate installation position to optimize the distance $r$ can effectively optimize the aperture field distribution and improve the aperture efficiency. This optimization can be formulated as 
\begin{equation}
\begin{aligned}
\text{(P1)} \ \min_{ r } \  |2r\sin (\frac{\theta_h}{2})-L\cos \alpha|,
\end{aligned}
\end{equation}
where $\theta_h$ is denoted as the HPBW. 
Recall that $\theta_h = |\theta^{*}_1-\theta^{*}_2|$, we have the optimum $r_{\rm opt}$ as
\begin{equation}
r_{\rm opt} = \frac{L\cos \alpha}{2\sin (\frac{|\theta^{*}_1-\theta^{*}_2|}{2})},
\end{equation}
where  
$\theta^{*}_1=\sin ^{-1}\left(  1.391\times \frac{2}{N kd}+\sin \theta^{\text s}\right)$, and 
$\theta^{*}_2=\sin ^{-1}\left(  -1.391 \times \frac{2}{N kd}+\sin \theta^{\text s}\right)$.

\subsubsection{RIS aperture optimization}
Based on the relationship between beamwidth, the number of array units, and beam direction, we can also determine the required RIS aperture size (i.e., the number of units needed in a uniform array) for signal relays at a fixed location.

Assuming the size of each unit within the uniform linear array is $\lambda/2$, the relationship between the number of units 
$N$ and the RIS aperture $L$ is given by $L = N \cdot\lambda/2$. We can formulate the optimization problem as
\begin{equation}
\text{(P2)} \ \min_{ N \in \mathbb{Z}^+ } \  |2r\sin (\frac{\theta_h}{2})-L\cos \alpha|.
\end{equation}
The optimum can be achieved by 
\begin{equation}
N_{\rm opt} = {\rm round}(  \frac{4r\cdot \sin (\frac{|\theta^{*}_1-\theta^{*}_2|}{2})}{\lambda \cdot \cos \alpha} ).
\end{equation}
Here, ${\rm round(\cdot)} $ denotes the operator of mapping to the nearest whole number,
$\theta^{*}_1=\sin ^{-1}\left(  1.391 \times \frac{2}{N kd}+\sin \theta^{\text s}\right)$, and 
$\theta^{*}_2=\sin ^{-1}\left(  -1.391 \times \frac{2}{N kd}+\sin \theta^{\text s}\right)$.

Optimizing the position and aperture size of the RIS can effectively focus the beam and ensure uniform illumination of the RIS's effective aperture, thereby achieving high aperture efficiency and directional gains. We will verify this through simulations in Section~\ref{Section5}.

\subsection{Beam Optimization}

For a large-aperture RIS, if the position is fixed, a single beam is often insufficient to cover the entire aperture area (as shown in Fig.~\ref{F2-1}), especially in the near field. It becomes necessary to design the reflecting beam through optimization methods. In this context, traditional methods often struggle to achieve the desired beam coverage and efficiency due to issues such as side-lobe interference and complex implementation requirements. 

It has been demonstrated that the phase optimization methods to achieve more beam numbers may reduce the antenna directivity but has no significant impact on the width of the beam width\cite{nayeri2018reflectarray,xiong2023fairbeamallocationsreconfigurable}. To design desired radiation patterns for RISs, we propose a modified method based on the MA algorithm, whose convergence and optimality have been proven in our previous work~\cite{xiong2023fairbeamallocationsreconfigurable}. For simplicity and clarity, we assume all the RISs have the same number of units.


We adopt the scheme to optimize phase configuration RIS by RIS. Specifically, for RIS$\{1\}$, we optimize its phase configuration to achieve uniform coverage of the reflected beam across the aperture of RIS $\{2\}$, aligning the coverage area with that of RIS $\{2\}$, based on the signal expression \eqref{Model_SISO_M}. For RIS$\{k\}$ ($k\geq 2$), the analysis approach remains consistent, transitioning to the signal model as depicted in \eqref{RIS-RIS-UE}. 

Similarly, initially taking the two-hop RIS scenario as an example, the beam design can be classified into two types. If the width of a single beam is over or exact illumination on RIS$\{2\}$, as depicted in Fig.~\ref{F2-2} and Fig.~\ref{F2-3}, we set the center of RIS $\{2\}$ as the target beam direction, and directly obtain the required phase configuration for RIS $\{1\}$ by
\begin{equation}
\omega_{n,{\rm opt}} =  2 \pi (r^\text{i} (n)+r^\text{s} (n)) / \lambda, \ n = 1,\cdots, N. 
\end{equation}
Otherwise, for partial illuminations in Fig.~\ref{F2-1}, we put forward the following optimization steps for RIS beam design, termed the multi-beam approach.

The optimization begins with determining the sampling points on RIS $\{2\}$, which will serve as the target directions for RIS $\{1\}$ to generate multiple beams, thereby achieving a uniform illumination on RIS $\{2\}$. To avoid excessive computational complexity and potential unsolvable cases caused by too many sampling points, as well as overly narrow illuminated areas due to too few sampling points, we determine the points according to the relationship between the beamwidth and RIS aperture. The number of points can be determined through
\begin{equation}\label{Sampling}
Z = 2 \times \lceil \frac{L\cos \alpha}{2r\sin (\frac{\theta_h}{2})} \rceil -1,
\end{equation}
where $\lceil \cdot \rceil$ denotes the ceiling function.

Set the $Z$ points on RIS$\{2\}$ as the hypothetic UEs, we explore a multi-beam design strategy with the objective of maximizing the minimum power. Formally, this strategy can be expressed as 
\begin{equation}\label{FB1}
  \max_{ \omega_1, \dots, \omega_{N} \in [0, 2\pi) } \min_{ z } \  \left \{ P^\text{s} ( r^\text{s}_z, \theta^\text{s}, \phi^\text{s} ), z = 1, \ldots, Z \right \}
\end{equation}
Here $P^\text{s} (\cdot) = \lvert E^\text{s} (\cdot) \rvert^2$ represents the power at a specific location.
 
Recall the equation in \eqref{Model_SISO_M}, if we define $\mathbf{w} = [e^{j \omega_1}, \cdots, e^{j \omega_N}]^H$, and $\mathbf{h}  = [h_1 , \cdots, h_N ]^T$ with $h_n = \frac{ e^{ - j 2 \pi r^\text{i} (n) / \lambda} e^{ - j 2 \pi r^\text{s} (n) / \lambda } }{ r^\text{i} (n)  r^\text{s} (n) }$, the optimization can be reformulated as 
\begin{equation}\label{FB2}
\begin{aligned}
\text{(P3)} \ \max_{ \omega_1, \dots, \omega_{N} \in [0, 2\pi) } \min_{z} \  \left \{   \mathbf{w}^H \mathbf{H}_z  \mathbf{w} , z=1, \ldots, Z \right \},
\end{aligned}
\end{equation}
where $\mathbf{H} = \mathbf{h}  \mathbf{h}  ^H$.

For RIS$\{k\}$ with $2\leq k\leq K-1$, this optimization problem can also be formulated, with the distinction that $h_n = E^\text{s}_{k-1} (r^\text{s}_{k-1}(n), \theta^\text{s}_{k-1},\phi^\text{s}_{k-1}) \frac{e^{-j 2 \pi r^\text{s}_k(n )/ \lambda}}{r^\text{s}_k(n )}$, as in \eqref{RIS-RIS-UE}. 

Problem (P3) belongs to the class of semi-infinite max-min problems, which can be optimally solved by our previously proposed MA algorithm~\cite{xiong2023fairbeamallocationsreconfigurable}. 
\begin{algorithm}
\caption{multi-beam approach for RIS phase configuration}
\label{alg-1.0}
\begin{algorithmic}[1]
\Procedure {OP}{optimization on RIS phases for $k \in[1,K]$}
	\If {$k \leq K-1$}
	\State calculate the samp point $Z$ on RIS$\{k+1\}$ by \eqref{Sampling}
	\State solve \eqref{FB2} for $\mathbf{w}_{\rm opt}^k$ by MA method \cite{xiong2023fairbeamallocationsreconfigurable}
	\Else 
     \State calculate $\omega_{n,\rm {opt}}$ by \eqref{RISK}
	\State calculate $\mathbf{w}_{\rm opt}^k$ with $\omega_{n,\rm {opt}}$
	\EndIf
   \State \textbf{return} $\{\mathbf{w}^1_{\rm opt}, \mathbf{w}^2_{\rm opt}\cdots,\mathbf{w}^K_{\rm opt}\}$
   \EndProcedure
\end{algorithmic} 
\end{algorithm}

After $K$ hops, a single UE receives the final signal. As for the last RIS$\{K\}$, and the phase configuration can be obtained as
\begin{equation}\label{RISK}
\omega_{n, \rm {opt}}^k = \arg \{E^\text{s}_{k-1} (r^\text{s}_{k-1}(n), \theta^\text{s}_{k-1},\phi^\text{s}_{k-1})\}+
  2 \pi r^\text{s}_k(n)/ \lambda
\end{equation}
Thus, the phase of all $K$ RISs are all optimally configured. The multi-beam approach for partial illuminations can be summarized in Algorithm \ref{alg-1.0}

The proposed beam optimization method features high flexibility in the multi-hop transmission of the wireless signal. The multi-beam approach leverages not only the positions of the transmitter-receiver and RISs but also utilizes channel state information if available. The primary distinction between them lies in the elements within matrix $\mathbf{H}$ in optimization problem (P3).

\section{Simulation and Prototype Experiments}\label{Section5}

We now investigate the beamforming capabilities of RISs to demonstrate the effectiveness of the derived methods and optimization strategies in multi-hop RIS-assisted communications. Initially, we validate the relationship between beamwidth, beam direction, and the number of units based on the derived HPBW expressions. Subsequently, we analyzed the aperture efficiency under different RIS deployment strategies, confirming that the optimized deployments are effective.
\begin{figure}[htbp]
  \centering
  \subfigure[]{
  \label{BW-S}
  \includegraphics[width=0.75\linewidth]{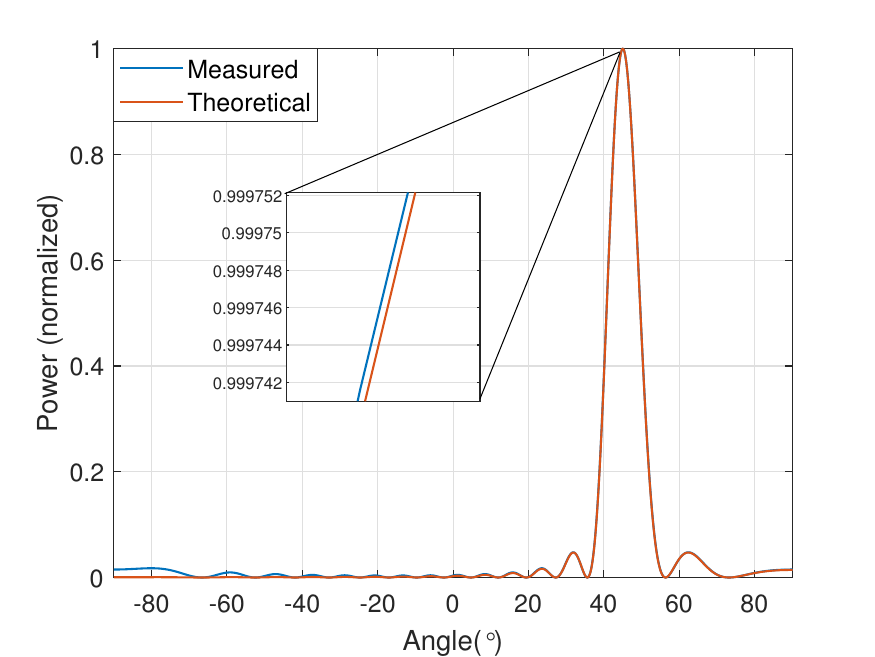}}
  \subfigure[]{
  \label{BW-M}
  \includegraphics[width=0.75\linewidth]{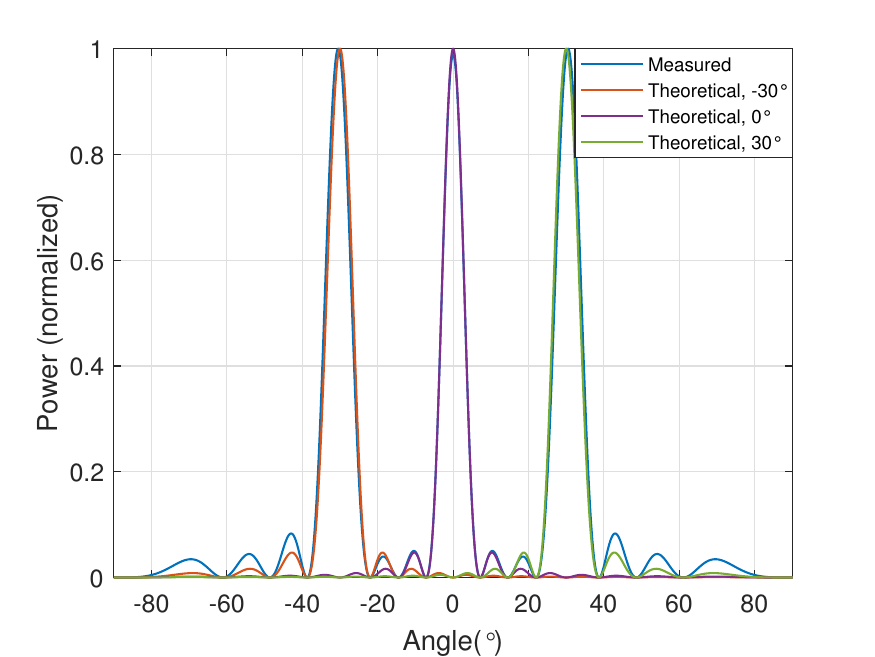}}
  \caption{Comparison of theoretical and measured half-power beamwidth. (a) Single beam. (b) Multiple beams.}
  \label{Beamwidth}
\end{figure}

Furthermore, we evaluate the impact of different aperture sizes and distances between RISs on the received signal power by employing the proposed multi-beam method and traditional single-beam method. The simulation and prototype-experiment outcomes highlight the importance of the proposed strategies for enhancing communication quality in multi-hop RIS-assisted communications. In all conducted simulations, the coordinate of the Tx remains fixed at (0, 0, 3).

\subsection{Beamwidth of the Reflecting Wave}

In Fig~\ref{Beamwidth}, we compare the radiation patterns derived from array factor theory with those obtained through simulation tests. We employed a linear RIS array with 16 units. The results in Fig.~\ref{BW-S} indicate that at a beam direction of 45°, the derived array factor closely matches the actual beam radiation pattern. Additionally, when using the MA algorithm to generate multiple beams, the main beams at $-30^\circ$, $0^\circ$, and $30^\circ$ all align with the derived, as shown in Fig.~\ref{BW-M}.

To further validate the relationship between the derived beamwidth, beam direction, and the number of units, we conduct experiments with different parameter settings and compare the theoretical beamwidth calculations with the actual beamwidths. The outcomes are illustrated in Table~\ref{beamwidth}. It can be observed that the HPBW becomes smaller with an increasing number of units. Additionally, the beam direction influences the reflected beamwidth, with beams at smaller elevation angles exhibiting narrower widths. In summary, the derived beamwidth in this work aligns closely with the actual measurements. When the number of units exceeds 16, the discrepancy can be less than $0.02^\circ$.

\begin{table}[!htbp]
\caption{Half-power beamwidth correspond to uniform linear RISs.} \label{beamwidth}
\centering
\setlength{\tabcolsep}{0.8mm}
\setstretch{1.1} 
\begin{tabular}{ccccc}
\toprule[1.5pt]
\multirowcell{2}{RIS units\\} & \multirowcell{2}{Beam direction \\$\theta$} &  \multirowcell{2}{Theoretical \\ HPBW}&  \multirowcell{2}{Measured \\ HPBW} & \multirowcell{2}{Gap}\\ \\
\midrule
8    & $30^\circ$  & $14.73^\circ$  &$14.84^\circ$&  $0.11^\circ$\\
8    & $45^\circ$  &$18.25^\circ$   &$18.38^\circ$&  $0.13^\circ$\\
8    & $60^\circ$  &$28.56^\circ$   &$28.85^\circ$&  $0.29 ^\circ$  \\
16   & $30^\circ$  & $7.33^{\circ} $          &$7.2^\circ$  &   $0.13^\circ$\\
16   & $45^\circ$  & $ 9.01^{\circ}$          &$9.02^\circ$ & $ 0.01^\circ$\\
16   & $60^\circ$  &$12.97^\circ$   &$12.99^\circ$&  $0.02^\circ$\\
32   & $30^\circ$  & $3.66^\circ$   &$3.67^\circ$ &  $0.01^\circ$\\
32   & $45^\circ$  & $4.49^{\circ}  $         & $4.49 ^\circ $      & $ 0^\circ$\\
32   & $60^\circ$  & $ 6.38^{\circ} $         & $6.38  ^\circ$      & $0^\circ$ \\
64   & $30^\circ$  &  $1.83^{\circ} $         &$1.83^\circ$ &  $0^\circ$\\
64   & $45^\circ$  & $ 2.24 ^{\circ}$         & $2.24^\circ $       & $ 0^\circ$\\
64   & $60^\circ$  &$3.18^\circ$    &$3.18 ^\circ $       & $0^\circ$\\
64   & $75^\circ$  &$6.26^\circ$    &$6.27  ^\circ$      & $0.01 ^\circ$\\
\bottomrule[1.5pt] 
\end{tabular}
\end{table}

\subsection{Aperture Efficiency in Deployment Optimization}
Deployment optimization in this work primarily targets the enhancement of aperture efficiency. By adjusting the current RIS's aperture size (number of units) or its distance from the preceding RIS based on the width of the beam reflected by the preceding RIS, higher aperture efficiency can be achieved in multi-hop setups.
For comparison, we defined the ratio of RIS effective aperture to beam illumination (EA-B ratio) in either linear or planar array as
\begin{equation}
{ABR} = \frac{L\cos \alpha}{2r\sin (\frac{\theta_h}{2})}.
\end{equation}
When the EA-B ratio \textbf{nears} $\mathbf{1}$, the theoretical aperture efficiency and achievable directional gain are maximized. We employ two successive RISs for illustration.

\begin{figure}[htbp]
\centering
\includegraphics[width=.75\linewidth]{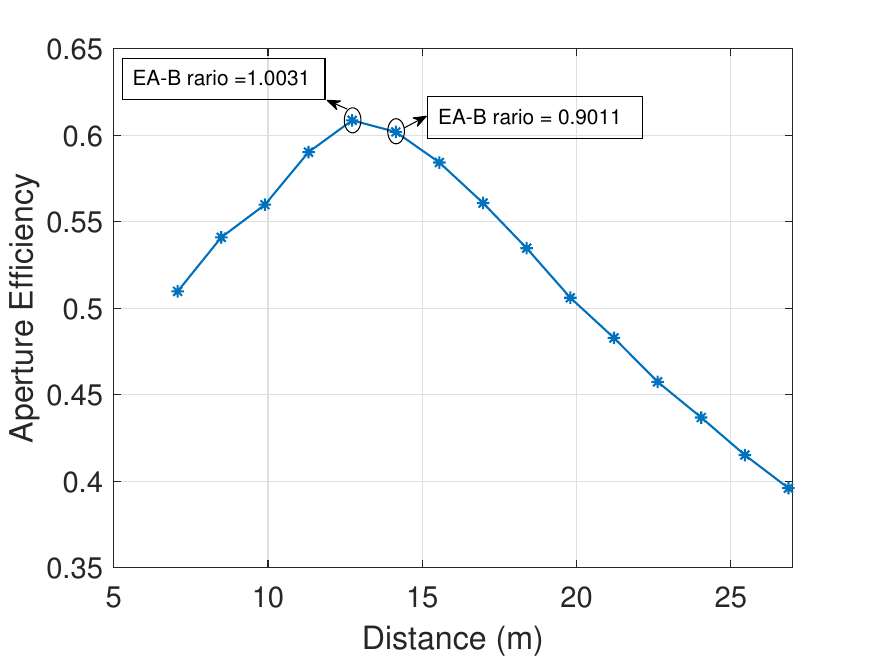}
\caption{Aperture efficiency as the function of distance between RIS$\{1\}$ and RIS $\{2\}$, with fixed RIS units $32\times 32$, and reflecting elevation angle $\theta = 45^\circ$.}
\label{ApertureEfficiencyWithDistance}
\end{figure}

\begin{figure}[htbp]
\centering
\includegraphics[width=.75\linewidth]{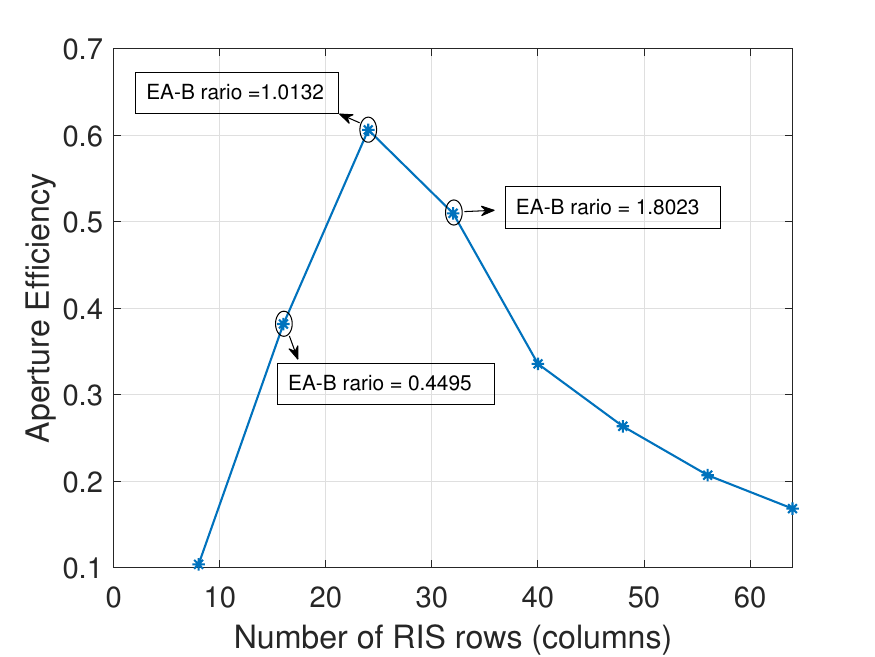}
\caption{Aperture efficiency as the function of RIS units, with RIS$\{1\}$ and RIS $\{2\}$ positioned (0, 0, 0) and (5, 0, 5).}
\label{ApertureEfficiencyWithN}
\end{figure}
\subsubsection{RIS position optimization for the problem {\rm(P1)}}

We initially simulate a scenario where the number of RIS units is fixed at $32 \times 32$, maintaining a constant reflection beam angle between RIS$\{1\}$ and RIS$\{2\}$ to ensure an unchanged beamwidth. By adjusting the coordinates of RIS$\{2\}$, we observe and record its aperture efficiency. 

The results in Fig.~\ref{ApertureEfficiencyWithDistance} show that as the distance between the two RISs increases, RIS$\{2\}$'s aperture efficiency initially rises and then declines. The specific values in Table.~\ref{ApertureEfficiency} indicate that when RIS$\{2\}$ is positioned at (9, 0, 9), the aperture efficiency is maximized, with the corresponding EA-B ratio being $1.0013$. Compared to its position at (19, 0, 19), RIS$\{2\}$ achieves an additional power gain of nearly 2 dB.

\subsubsection{RIS aperture optimization for the problem {\rm (P2)}}
In addition to the positions, we also evaluate the impact of RIS aperture size on the aperture efficiency. In this test, we fix the positions of RIS$\{1\}$ and RIS$\{2\}$ at (0, 0, 0) and (5, 0, 5), respectively, and measure the changes in aperture efficiency by simultaneously increasing the number of rows and columns (aperture size) of RISs. As shown in Fig.~\ref{ApertureEfficiencyWithN}, the results indicate that when the number of RIS units is $24\times24$, RIS$\{2\}$ achieves a maximum aperture efficiency of $0.6062$, where the corresponding EA-B ratio equals $1.0132$.

The two experiments on RIS deployment demonstrate that changing the position and aperture size of RISs significantly impacts the aperture efficiency and directional gain of RISs. This underscores the practical significance of optimizing RIS deployment in real-world scenarios.

\begin{table*}[!htbp]
\caption{Aperture efficiency and Direction Gain of RIS$\{2\}$, while RIS $\{1\}$ located at $(0,0,0)$.} \label{ApertureEfficiency}
\centering
\setlength{\tabcolsep}{1.2mm}
\setstretch{1.1} 
\begin{tabular}{ccccc}
\toprule[1.5pt]
RIS$\{2\}$ coordinate &  Inner-RIS Distance (m)&  EA-B Ratio   & Aperture efficiency &Direction gain (dBi)\\
\midrule
(5, 0, 5)    & 7.07  & $1.8023$  &0.5097&  32.16\\
(6, 0, 6)    & 8.49  &$1.5019$   &0.5409&  32.41\\
(7, 0, 7)    & 9.90  &$1.2874$   &0.5598&  32.56   \\
(8, 0, 8)   & 11.31  & 1.1264           &0.5902  &  32.79\\
$\mathbf{(9,\ 0,\ 9)}$  & $\mathbf{12.73}$  &  $\mathbf{1.0013}$         &$\mathbf{0.6086}$ &  $\mathbf{32.93}$\\
(10, 0, 10)  & 14.14  &$0.9011$   &0.6019& 32.88\\
(11, 0, 11)   & 15.56  & $0.8192$   &0.5842 &  32.75\\
(13, 0, 13)   & 18.38  & 0.6932         & 0.5346        & 32.36\\
(15, 0, 15)   & 21.21  &  0.6008          & 0.4828        &  31.92\\
(17,0,17)   & 24.04  &  0.5301         & 0.4367       &  31.48\\
(19,0,19)   & 26.87  & 0.4743         & 0.3959        &  31.06\\
\bottomrule[1.5pt] 
\end{tabular}
\end{table*}

\subsection{Aperture Field Distribution in Beam Optimization}

\begin{figure}[htbp] 
  \centering
  \subfigure[]{
  \label{S505}
  \includegraphics[width=.475\columnwidth]{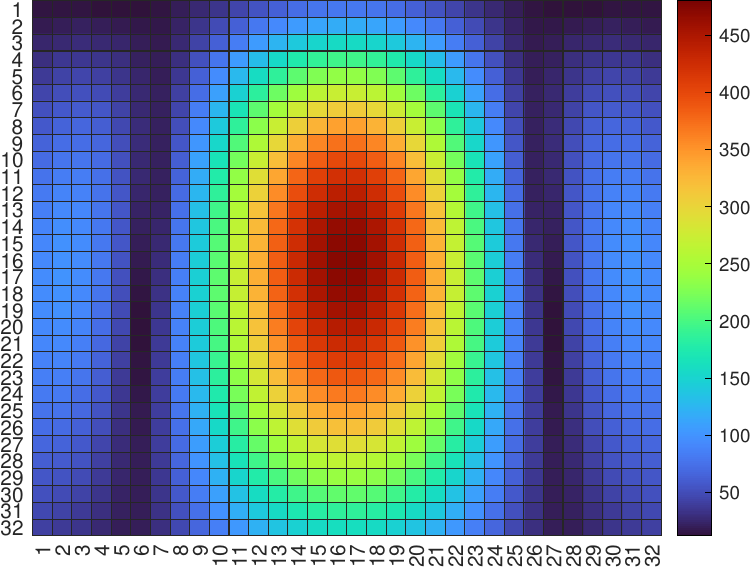}}
 \subfigure[]{
  \label{M505}
  \includegraphics[width=.475\columnwidth]{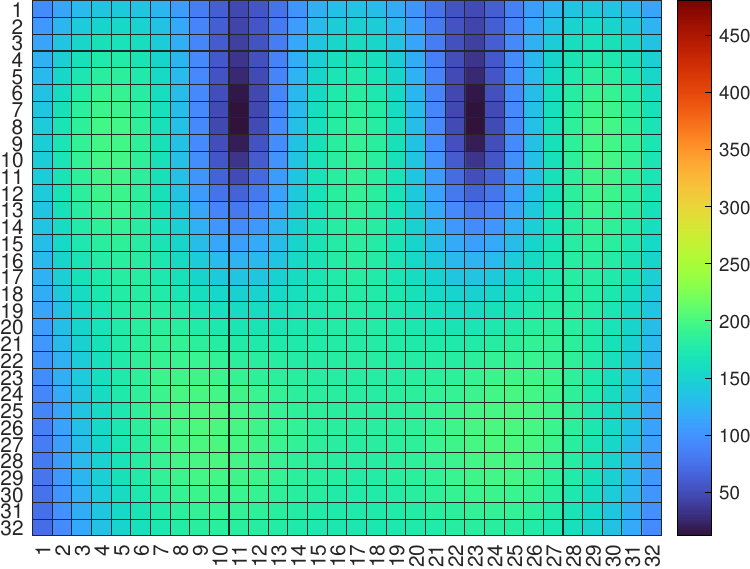}}
 \subfigure[]{
  \label{S404}
  \includegraphics[width=.475\columnwidth]{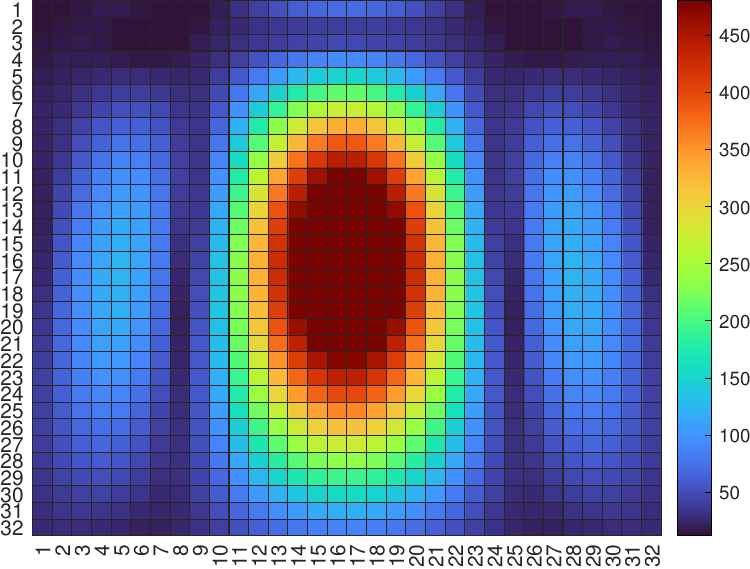}}
 \subfigure[]{
  \label{M404}
  \includegraphics[width=.475\columnwidth]{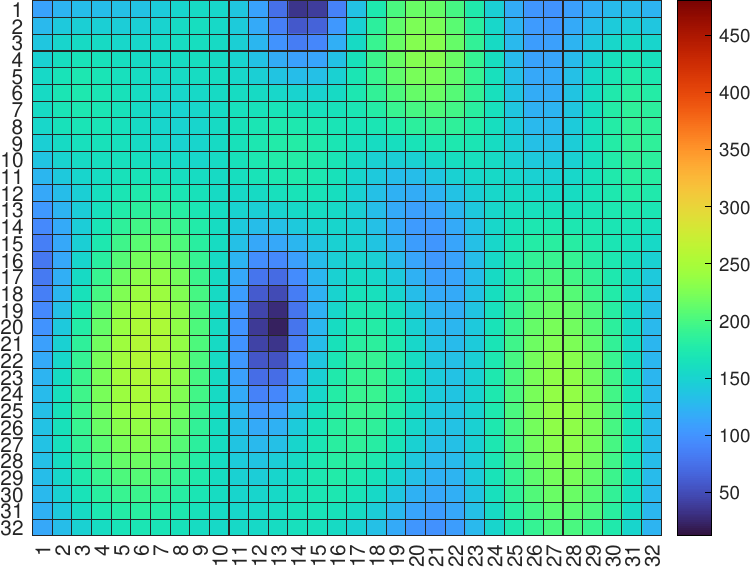}}
  \caption{The aperture field distribution of RIS$\{2\}$ in two-hop RIS-assisted communication, RIS $\{1\}$ are located at $(0,0,0)$, and $32 \times 32$ array units are employed for each RIS. (a) Single-beam scheme, with RIS $\{2\}$ located at $(5,0,5)$. (b) Multi-beam scheme, with RIS $\{2\}$ located at $(5,0,5)$. (c) Single-beam scheme, with RIS $\{2\}$ located at $(4,0,4)$. (d) Multi-beam scheme, with RIS $\{2\}$ located at $(4,0,4)$.}
\label{ApertureField}
\end{figure}
\begin{figure}[htbp]
  \centering
  \includegraphics[width=0.7\linewidth]{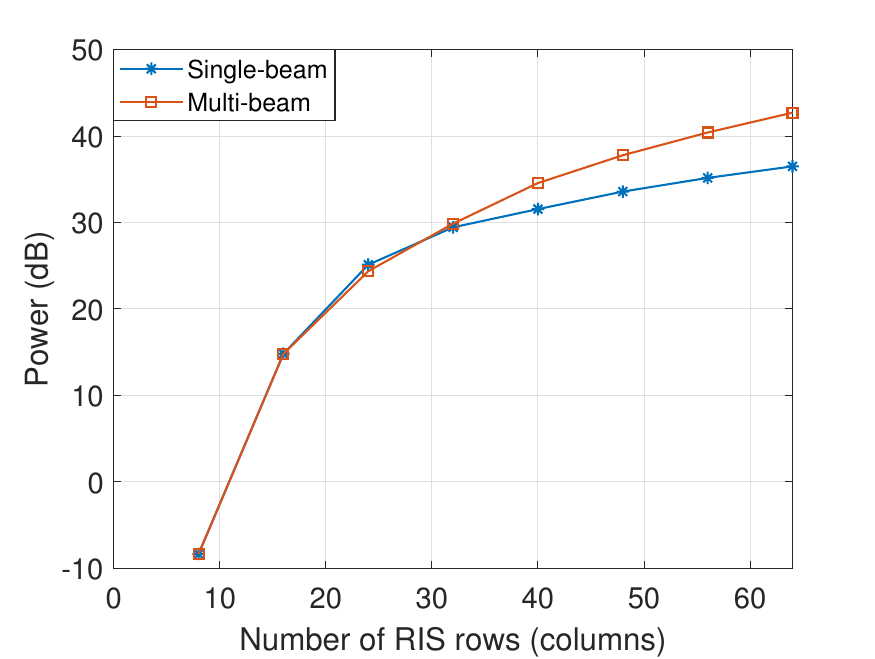}
  \caption{The received signal power at the final UE after two-hop RIS-assisted communication, as a function of the number of RIS units.}
  \label{PowerWithN}
\end{figure}

\begin{figure}[htbp]
  \centering
  \includegraphics[width=0.7\linewidth]{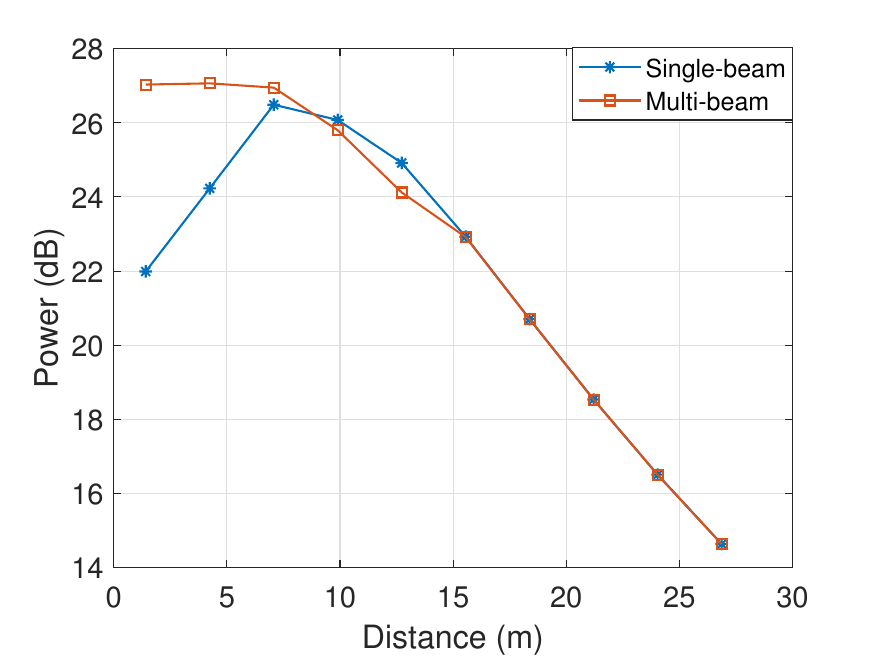}
  \caption{The received signal power at the final UE after two-hop RIS-assisted communication, as a function of the number of distances between the two RISs. The number of units remains $32\times 32$.}
  \label{PowerWithDistance}
\end{figure}

In practical scenarios where the aperture size and position of RISs are fixed and immutable, optimizing the beams remains crucial for achieving higher gains on received power, particularly in near-field environments. Based on the proposed signal model, we compare the effects on received signal power between traditional direct calculation (single-beam) methods~\cite{ma2023multi,mei2020cooperative} and the proposed multi-beam approach in partial illumination scenarios.

As aforementioned, the advantage of the multi-beam strategy lies in its ability to homogenize the aperture field distribution of RISs. To illustrate this, we generate heat maps illustrating the field distributions corresponding to both single-beam and multi-beam methods at two distinct positions of RIS$\{2\}$. Figs.~\ref{S505} and ~\ref{S404} demonstrate that using traditional single-beam methods leads to signal energy over-concentration on the central portion of the RIS board, intensifying as the distance between RISs decreases. In contrast, the multi-beam approach effectively alleviates this issue, by concurrently generating beam coverage across various regions of RIS$\{2\}$ using RIS$\{1\}$, as indicated in Figs.~\ref{M505} and \ref{M404}.



%
%
%

\subsection{Different Aperture Field Distribution Lead to Various Received Power}

To further investigate the impact of field distribution variations on received signal power, we conduct experiments comparing the multi-beam and single-beam approaches in two-hop RIS-assisted communication. The experiments consist of two primary sets: the first set observes the effect of beam optimization on received signal power with varying numbers of RIS units, while the second set examines the received signal power at different distances between RISs.

We first position Tx, Rx, and two RISs at (0, 0, 3), (5, 0, 0), (0, 0, 0), and (5, 0, 5) respectively, and varied the number of rows (columns) of the RIS. As illustrated in Fig.~\ref{PowerWithN}, when the number of rows (columns) is relatively small, the received signal power corresponding to both the single-beam and multi-beam approaches is comparable. As the number increases, the single-beam shows an advantage within a certain range. However, when the number exceeds a certain threshold (e.g., $32\times32$), the advantage of the multi-beam approach becomes increasingly evident. 

\begin{figure*}[t]
\centering
\includegraphics[width=0.85\linewidth]{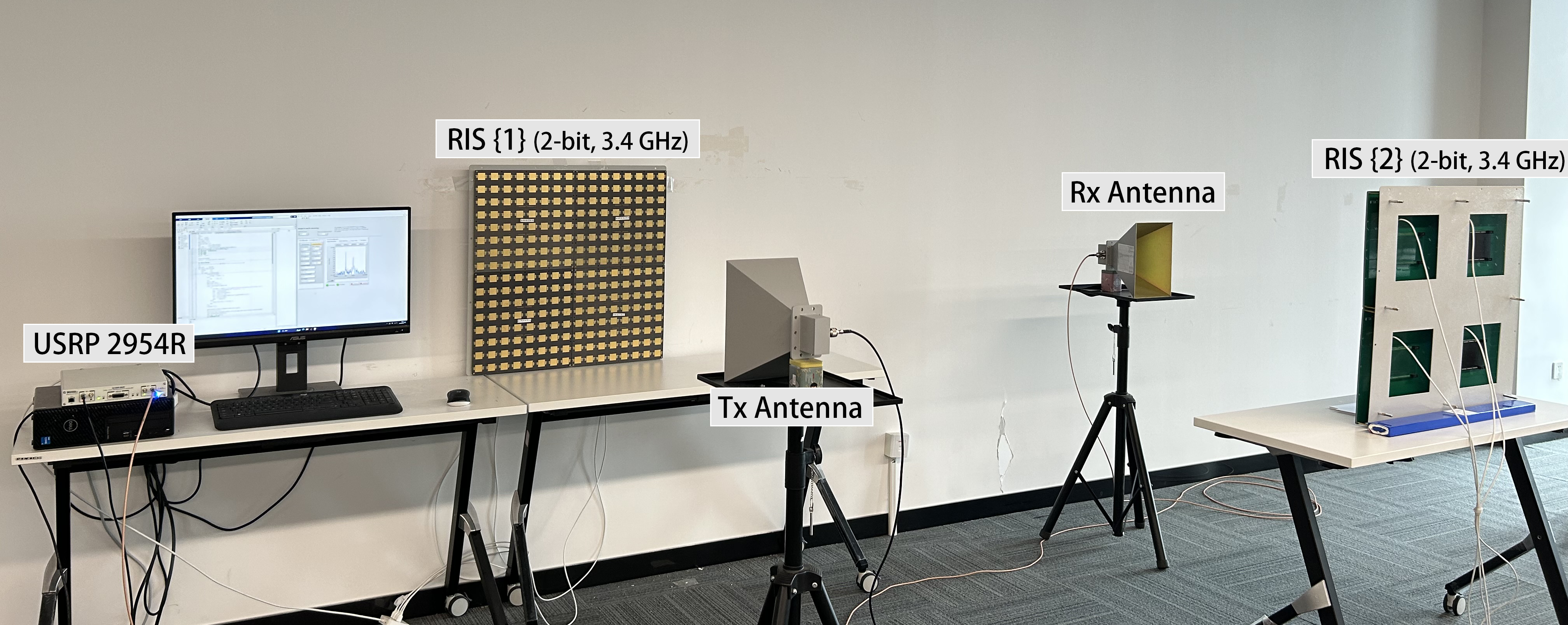}
\caption{The prototype of 2-bit RIS-assisted multi-hop wireless communication system.}
\label{system}
\end{figure*}
These findings can be attributed to the fact that when the number of RIS units is small, the beam's minimum resolution is insufficient, leading to over-illumination of the RIS aperture (as in Fig.~\ref{F2-2}). This results in similar aperture field distributions with the multi-beam and single-beam approaches. As the number of units gradually increases, the beam width and aperture size become closer (EA-B ratio approaches 1), the single-beam approach produces more regular and uniform beams, thereby slightly outperforming the multi-beam approach. However, with a further increase in the aperture size of the RIS, the multi-beam approach exhibits advantages in field distribution and demonstrates increasingly superior performance.

In another distance experiment, we maintain the positions of Tx and RIS$\{1\}$ constant and changed the position of Rx to (10, 0, 0). By adjusting the coordinates of RIS$\{2\}$ and its distance from RIS$\{1\}$, we observed the received signal power. The results in Fig.~\ref{PowerWithDistance} indicate that when the distance between RISs is relatively short, the multi-beam approach exhibits superior performance regarding received signal power. For instance, at a distance of 1.41 meters between RISs, this gain exceeds 5 dB.

From these two sets of experiments, it can be concluded that the multi-beam approach yields significant performance gains in partial illumination scenarios, validating the effectiveness of our proposed beam optimization strategy. Additionally, the results from Figs.~\ref{PowerWithN} and \ref{PowerWithDistance} also indicate that selecting the appropriate number of units and RIS positions is crucial for enhancing the received signal power, as highlighted in the deployment optimization.

%

\subsection{Experimental Results with Prototype System}

To assess the effectiveness of the proposed beam optimization method in practical settings, we conduct experiments with an operational two-hop RIS-assisted communication system. Within the framework of the present model, the phase disparities among different paths are computed directly based on the positions and configurations of the units. This characteristic facilitates beamforming without the necessity for explicit channel estimation procedures.


The prototype validations are conducted utilizing a two-hop 2-bit RIS-assisted wireless system operating at the central frequency of 3.4 GHz, as illustrated in Fig.~\ref{system}. Each RIS is composed of $16 \times 16$ units, with each unit carrying two PIN diodes, corresponding four different response combination when biased in forward and reverse respectively. The source signal with baseband frequency at 100 kHz is generated and modulated using linear frequency modulation (LFM) by employing the USRP 2954R. The signal is transmitted through a transmitting horn antenna (Tx), reflected by the two-hop RISs, and ultimately received by a receiving horn (Rx). 

\begin{figure}[htbp]
\centering
\includegraphics[width=0.85\linewidth]{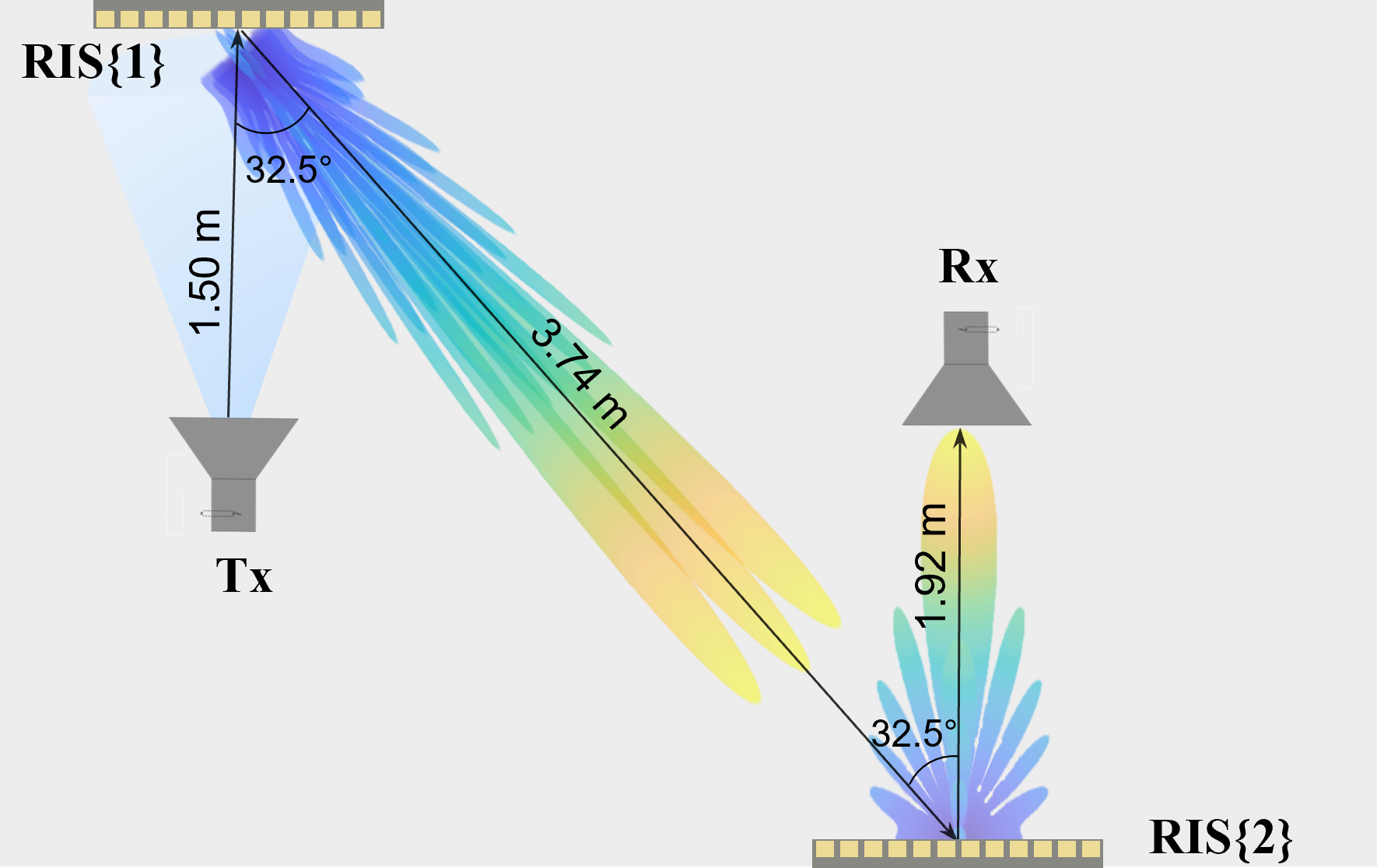}
\caption{Experimental setup in multi-hop RIS-assisted communication.}
\label{Setup}
\end{figure}

\begin{figure}[htbp] 
  \centering
  \subfigure[]{
  \label{SRIS1}
  \includegraphics[width=.45\columnwidth]{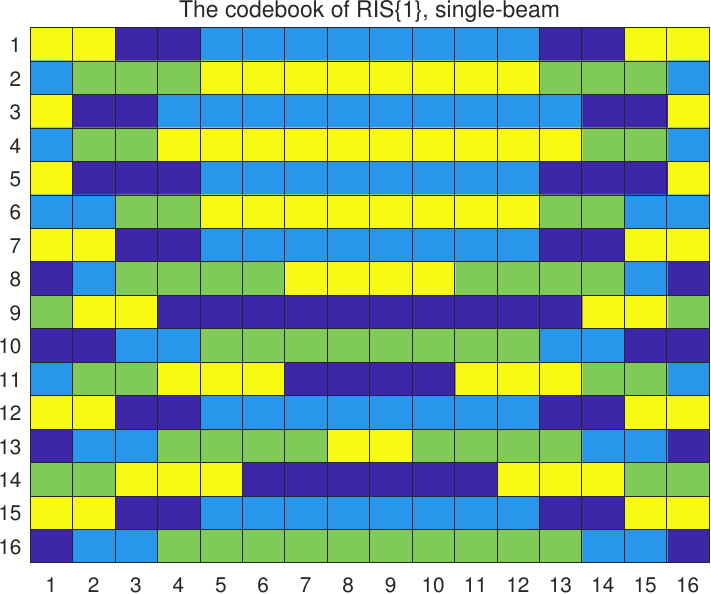}}
 \subfigure[]{
  \label{SRIS2}
  \includegraphics[width=.45\columnwidth]{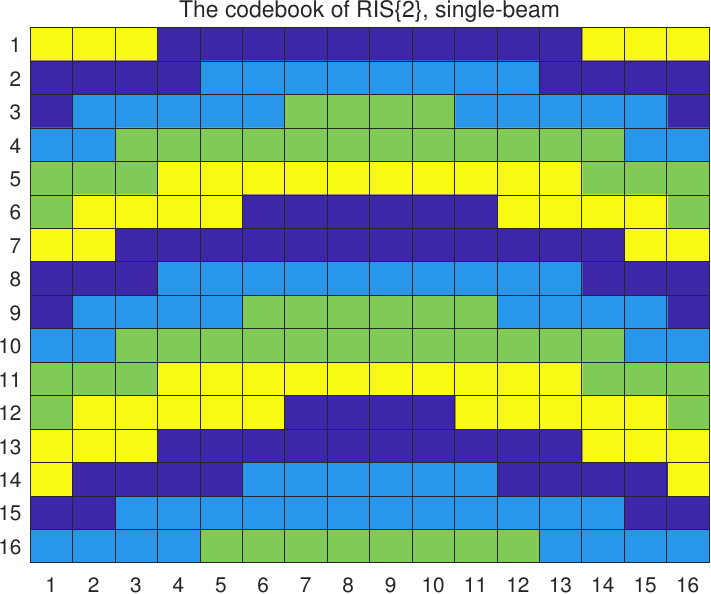}}
 \subfigure[]{
  \label{MRIS1}
  \includegraphics[width=.45\columnwidth]{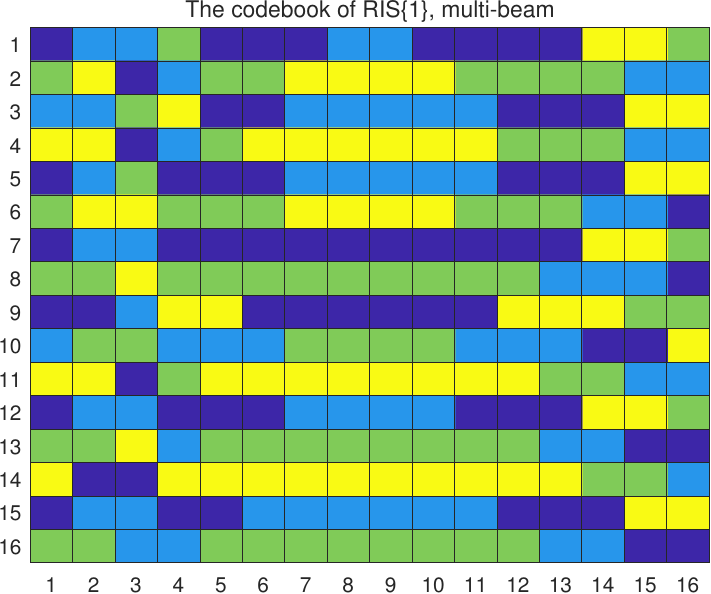}}
 \subfigure[]{
  \label{MRIS2}
  \includegraphics[width=.45\columnwidth]{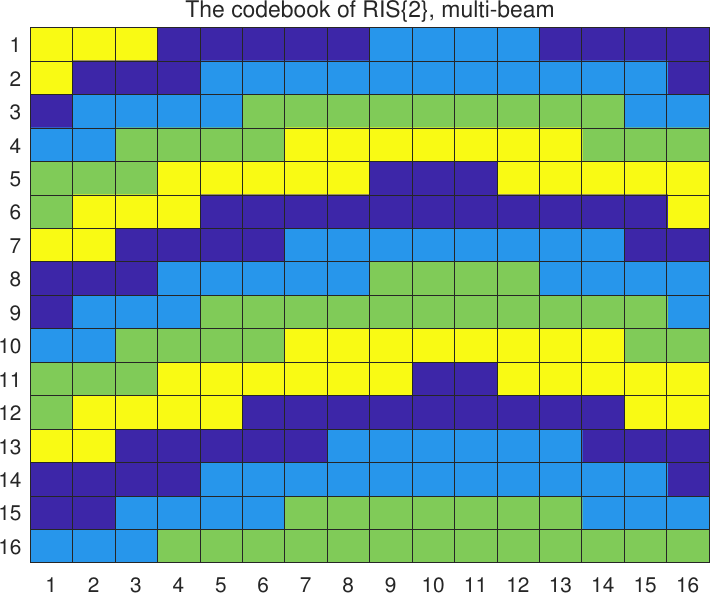}}
  \caption{The control codebooks correspond to the phase configurations of a $16 \times 16$ RIS, derived through various approaches at the operating frequency of 3.4 GHz. (a) RIS $\{1\}$ in the single beam scheme. (b) RIS $\{2\}$ in the single beam scheme. (c) RIS $\{1\}$ in the multi-beam scheme. (d) RIS $\{2\}$ in the multi-beam scheme. }
\label{codebooks}
\end{figure}

\begin{figure}[htbp]
  \centering
  \subfigure[]{
  \label{Not-configured}
  \includegraphics[width=0.7\linewidth]{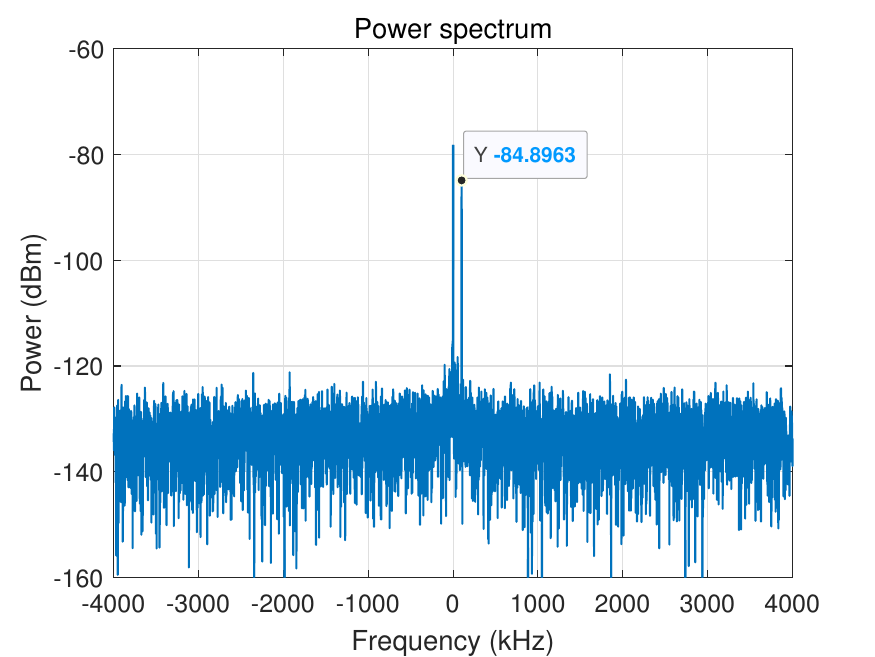}}
  \subfigure[]{
  \label{Single}
  \includegraphics[width=.7\linewidth]{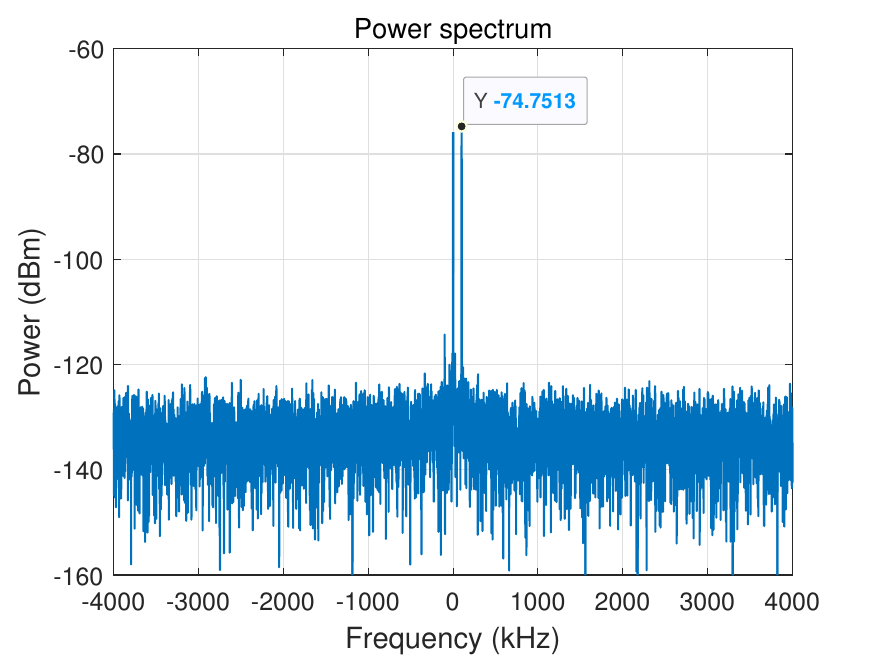}}
  \subfigure[]{
  \label{Multi}
  \includegraphics[width=.7\linewidth]{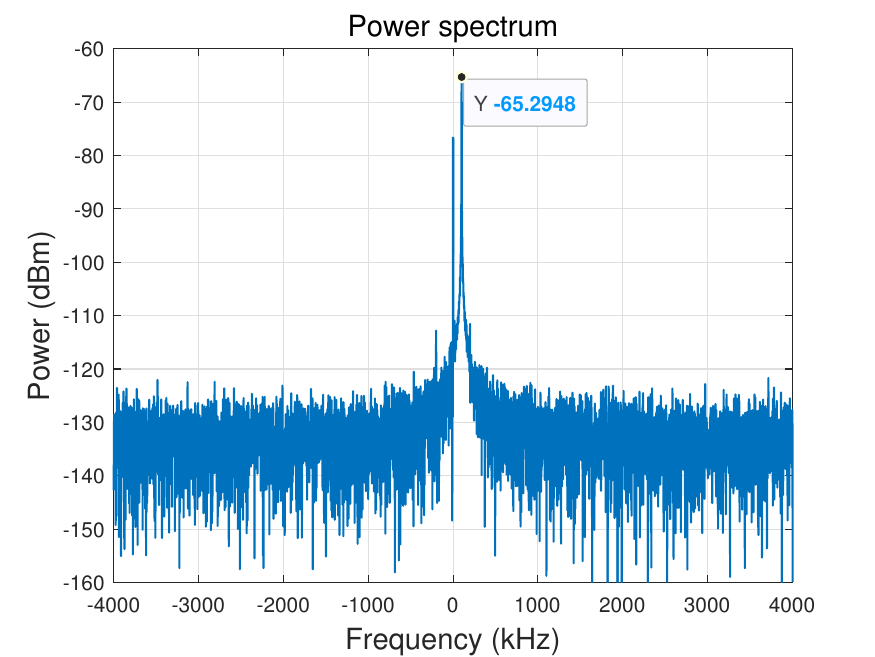}}
  \caption{The actual received power observations with RIS phase under (a) not-configured. (b) optimized through the single-beam scheme. (c) optimized through the multi-beam scheme.}
  \label{experiment}
\end{figure}

We initially measure the received signal power with no specific phase configuration. The Tx and Rx antenna is positioned in front of the RIS$\{1\}$ and RIS$\{2\}$ with a distance of $1.50$ m and $1.92$ m, respectively. The distance between the two RIS centers is set at $3.74$ meters, with a reflection angle of $32.5^\circ$, as shown in Fig.~\ref{Setup}. All the devices are at the same horizontal height. The observed power spectrum is illustrated in Fig.~\ref{Not-configured}, it shows that the received power level is approximately -84.90 dBm. 

For each of the comparative objectives, we perform different methods to obtain the phase configurations on RIS$\{1\}$ and RIS$\{2\}$ based on the proposed signal model. These configurations are then transformed into control codebooks. As shown in Fig.\ref{codebooks}. Four different colors represent four distinct phase configurations with intervals of $\pi/2$. The codebooks in Fig.~\ref{SRIS1} and Fig.~\ref{SRIS2} are obtained through phase compensation methods (belongs to single-beam method) according to the locations of the equipment~\cite{stutzman2012antenna}, while Fig.~\ref{MRIS1} and ~\ref{MRIS2} are from the proposed Algorithm~\ref{alg-1.0}.

The experimental results in Fig.\ref{experiment} indicate that, although we can achieve the optimal phase configuration for the two RIS units in the single-beam scheme, the resulting signal gain is not substantial, only around 10 dB, due to the beam illuminating an area far smaller than the RIS aperture. However, when considering the aperture efficiency of the RIS and employing a multi-beam approach to increase the illuminated area of the RIS, the achieved signal power gain is further enhanced, surpassing 19 dB compared to the unconfigured RIS phase. These outcomes underscore the effectiveness of the proposed beam optimization strategy in multi-hop RIS-assisted communications.

\section{Conclusion and Future Directions}\label{Section6}


This paper focuses on the design and optimization of multi-RIS-assisted communications, particularly in multi-hop scenarios. As a reflective aperture antenna, the aperture efficiency of RIS determines its maximum directional gain, a factor often overlooked in current beamforming schemes. This study examines the field distribution on the RIS aperture, derives a closed-form expression for the half-power beamwidth of beams reflected by RIS, and elucidates the relationship between beamwidth, RIS aperture size (number of units), and beam direction. Drawing from beamwidth and aperture considerations, the paper proposes deployment and beam optimization strategies to enhance the performance of multi-hop RIS-assisted communications without the requirement of any CSI estimation. Simulation and prototype experiments validate the accuracy of the derived results and the effectiveness of the proposed strategies. These optimization frameworks and methods are applicable across a broad spectrum of multi-hop RIS networks, including communications in tunnels and urban environments with dense blockages/obstructions.

\appendices
\section{Proof of Equations}\label{AppA}
For the equation
\[
\begin{aligned}
 \mathrm{AF}(\theta)=\sum_{n=1}^N e^{j(n-1) \psi},
\end{aligned}
\]

Multiplying both sides by $e^{j\psi}$, it can be written as
\[
(\mathrm{AF}(\theta)) e^{j \psi}=e^{j \psi}+e^{j 2 \psi}+e^{j 3 \psi}+\cdots+e^{j(N-1) \psi}+e^{j N \psi}
\]

Subtracting the last from the former reduces to
\begin{equation}
\mathrm{AF}(\theta)\left(e^{j \psi}-1\right)=\left(-1+e^{j N \psi}\right)
\end{equation}
which can also be written as
\begin{equation}
\begin{aligned}
\mathrm{AF}(\theta) & =\left[\frac{e^{j N \psi}-1}{e^{j \psi}-1}\right]=e^{j((N-1) / 2] \psi}\left[\frac{e^{j(N / 2) \psi}-e^{-j(N / 2) \psi}}{e^{j(1 / 2) \psi}-e^{-j(1 / 2) \psi}}\right] \\
& =e^{j[(N-1) / 2] \psi}\left[\frac{\sin \left(\frac{N}{2} \psi\right)}{\sin \left(\frac{1}{2} \psi\right)}\right]
\end{aligned}
\end{equation}

In fact, if the reference point is the physical center of the array, the array factor reduces to
\begin{equation}
\mathrm{AF}(\theta)=\left[\frac{\sin \left(\frac{N}{2} \psi\right)}{\sin \left(\frac{1}{2} \psi\right)}\right]
\end{equation}

which can be normalized and approximated as 
\[
\mathrm{AF}(\theta) \simeq\left[ \frac{\sin \left(\frac{N}{2} \psi\right)}{\frac{N}{2} \psi}\right]
\]

\section{}\label{AppB}

\begin{figure}[htbp]
\centering
\includegraphics[width=0.75\linewidth]{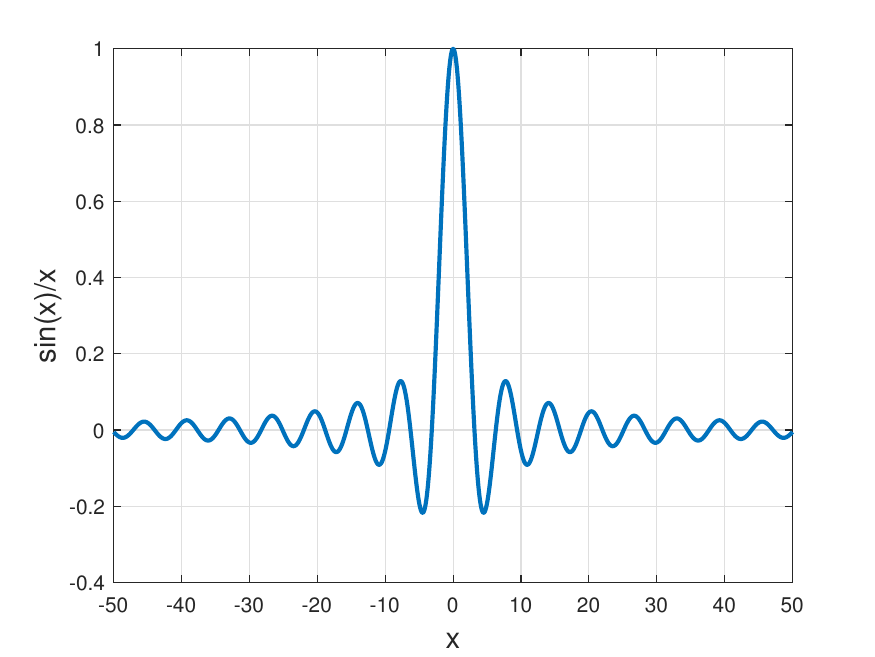}
\caption{Plot of $f(x) = \frac{\sin x}{x}$ function}
\label{x}
\end{figure}

Plot of function $f(x) = \frac{\sin x}{x}$, and $f(x,y) = \frac{\sin x}{x}\frac{\sin y}{y}$.
\begin{figure}[htbp]
\centering
\includegraphics[width=0.75\linewidth]{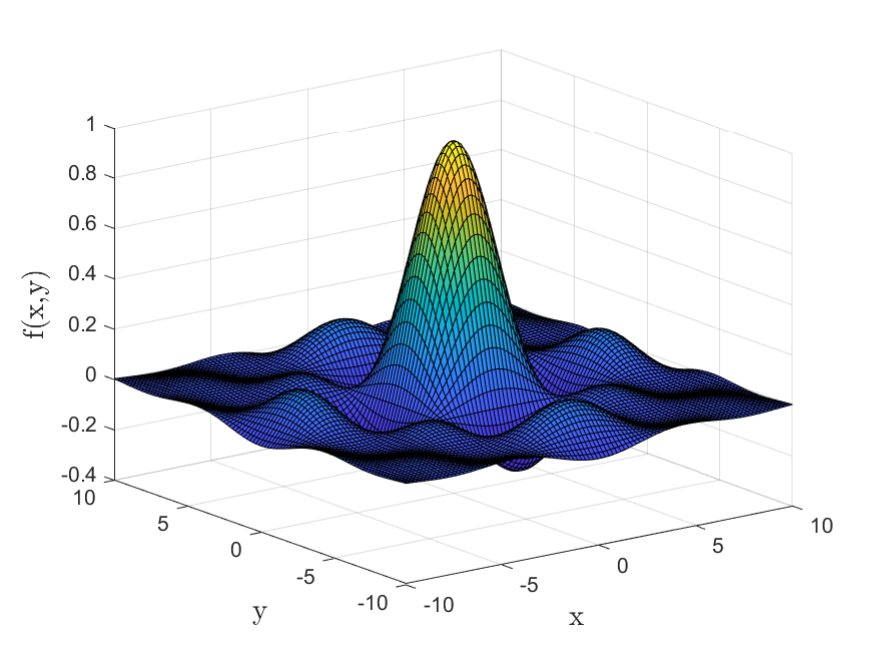}
\caption{Plot of $f(x,y) = \frac{\sin x}{x}\frac{\sin y}{y}$ function}
\label{xy}
\end{figure}


%

\bibliographystyle{IEEEtran}
\bibliography{Reference}
%
%
%
%
%
%
%
%
%
%

\newpage

 


\vspace{11pt}


\vfill
\end{document}